\documentclass[a4paper,12pt]{article}

\usepackage[latin1]{inputenc}
\usepackage{amsmath}
\usepackage{amsmath, amsthm}
\usepackage{amsfonts,amssymb}
\usepackage{mathrsfs}
\usepackage{color}
\theoremstyle{plain}

\usepackage[top=3.7cm, bottom=3.8cm, left=2cm , right=2cm]{geometry}%Ajustement des marges

\newtheorem{theo}{Theorem}[section]
\newtheorem{req}{Remark}[section]
\newtheorem{pp}[theo]{Proposition}

\newtheorem{df}{Definition}[section]

\newtheorem{cor}[theo]{Corollary}
\newtheorem{lem}[theo]{Lemma}
\newtheorem{hp}{Assumption}[section]

\title{Stochastic Utilities With a Given Benchmark Portfolio : Approach by Stochastic Flows\footnote{With the financial support of
the "Fondation du Risque"
and the F\'ed\'eration des banques Fran\c{c}aises.} \footnote{Key Words. forward utility, performance criteria, horizon-unbiased utility,
 consistent utility, progressive utility, portfolio optimization, optimal portfolio, duality, minimal martingale measure, Stochastic flows of homeomorphisms}}

\author{ El Karoui Nicole,
\thanks{ \small  LPMA, UMR CNRS  6632,  Universit\'e Pierre et Marie Curie, CMAP, UMR CNRS 7641, \'Ecole Polytechnique  }
\\
\and Mrad~Mohamed~ \thanks
{\small CMAP, UMR CNRS 7641, \'Ecole Polytechnique, \small  LAGA, UMR CNRS 7539,  Universit\'e Paris 13}}

\DeclareMathOperator*{\esssup}{ess\,sup}

\def\ind{{\mathbf{1}}}

\def\cal{\mathcal}
\def\P{{\mathbb P}}

\def\Q{{\mathbb Q}}
\def\R{{\mathbb R}}

\def\E{{\mathbb E}}

\def\F{{\cal F}}
\def\M{{\cal M}^e}

\def\Cc{{\cal C}}

\def\tU{{\tilde U}}
\def\tu{{\tilde u}}
\def\tv{{\tilde v}}
\def\tX{{\tilde X}}
\def\GX{{\mathscr X}}
\def\GY{{\mathscr Y}}
\def\PX{{\mathbb X}^+ }

\def\X{{\cal X}}

\def\Y{{\cal Y}}

\def\tU{{\tilde U}}

\def\eps{\varepsilon}

\newcommand{\rmi}{{\rm (i) $\>\>$}}

\newcommand{\rmii}{{\rm (ii) $\hspace{1.5mm}$}}

\begin{document}
 \maketitle
\abstract{The paper generalizes the construction by stochastic flows of  consistent
 utility processes introduced by M. Mrad and N. El Karoui in \cite{MRADNEK01}. The utilities random fields are defined from a
general class of processes denoted by $\GX$. Making minimal assumptions and convex constraints on test-processes, we
construct by composing two stochastic flows of homeomorphisms,  all the consistent stochastic utilities  whose the
optimal-benchmark  process is  given, strictly increasing in its initial condition. Proofs are essentially based on
stochastic change of variables techniques.
%\keywords{AAAAAAAAAAAAAAAAA}
}

\section{Introduction.}
The purpose of this paper is to generalize the construction of consistent utilities  by stochastic flows method introduced in
\cite{MRADNEK01}  in a It\^{o}'s framework  where securities are modeled as continuous It\^{o}'s semimartingales. The concept
of consistent stochastic utilities, also called "forward dynamic utilities", has  been introduced by M. Musiela and T.
Zariphopoulou in 2003 \cite{zar-03,zar-07} ; since this notion appears in the literature
in varied forms, in the work of T. Choulli, C. Stricker and L. Jia \cite{choulli},  V. Henderson and
D.  Hobson \cite{Hobson},  
F. Berrier, M. Tehranchi and Rogers \cite{Mike}, G. Zitkovic \cite{zitkovic} and in the work of M. Mrad and N. El Karoui in \cite{MRADNEK01}. Intuitively, a
stochastic utility should represent, possibly changing over time, individual preferences
of an agent. The agent's preferences  are affected over time by the information available on the market  represented by the filtration
$(\F_t, t \ge 0)$ defined on the probability space $(\Omega,\P, \F).$ 
For this, the agent
starts with today's specification of his utility, $u(0,x)=u(x)$ , and then builds the
process $U(t,x)$ for $t > 0$ taking into account the information flow given by $(\F_t, t \ge 0)$.
Consequently, its utility, denoted by $U (t, x)$ is a progressive process depending on
time and wealth, $t$ and $x$, which is as
a function of $x$ strictly increasing
and concave. 
In contrast to the classical literature, there is no pre-specified trading
 horizon at the end of which the utility datum is assigned. Consequently the initial function $U(0,x)$ is given in place of
$U(T,x)$ where $T$ is the time horizon in the classical problem . 
 These utility random fields will be called consistent progressive utilities in that follows.

Working on  a general framework,  our main contribution is the new approach by stochastic flows of consistent dynamic utilities,
proposed in Section \ref{NASF}. The idea is the same as in  \cite{MRADNEK01}:  
suppose the optimal process denoted by $X^*$ is strictly increasing
with respect to its initial capital. Denote by $\X$ the reverse flow of $X^*$ i.e. $\X_.(z):=(X^*_.)^{-1}(z)$, by $Y^*$ the
optimal process of the dual problem and by $U_x$ the first derivative of the random variable $U$ with respect  to the spacial
parameter $x$, then from the duality identity
$U_x(t,X^*_t(x))=Y^*(t,U_x(0,x))$   we easily get  $U_x(t,x)=Y^*\big(t,U_x(0,\X(t,x))\big)$ and finally  $U$  by simple 
integration. We then, by stochastic flows techniques,   construct
all consistent utilities generating $X^*$ as optimal process.\\

\noindent
Let us end this introduction by an overview of the paper. In the next section,  the framework is introduced   and the

definition of consistent stochastic  utilities is given. Also, in the next  section, the  model class $\GX$ of test-processes
is given and  a simple and intuitive example of  stochastic utility which  focuses on a sufficient assumption to the
existence of these random fields is developed.
Next  optimality conditions are established and   the question of duality  is elaborated. In paragraph \ref{TCN}, we show  the stability of the notion  of consistent utility by change of numeraire. 
Section \ref{NASF} is the core of the paper, we present our new approach.

%%%%%%%%%%%%%%%%%%%
\section{Consistent Stochastic Utilities}
 To get started, we consider a probability space $(\Omega,{\cal F},\mathbb{P})$, a time horizon $T_H\in (0,\infty]$
and a filtration $\mathbb F=({\cal F}_{0\le t \le T_H})$ satisfying the usual conditions of right-continuity and
completeness. Thus  only  the c\`adl\`ag version of $(\mathbb{P},\mathbb F)$-semimartingales are considered.

\noindent
Before  moving to the precise definition of the  utilities processes that is the subject of this work, the definition of the
notion of g-supermartingale is needed and will be extensively used.
\begin{df}
 A  stochastic process $(Z_t)_{t}$ will be called  a generalized supermartingale with respect to ${\cal F}$  if
 $\mathbb{E}(Z_t/{\cal F_s})\le Z_s$ whenever $s\ge0,~t\ge s$ with $\mathbb{E}(Z^+_t)<+\infty$ a.s. for any $t$, where
$Z^+_t$ is the positive part of $Z$ given by $Z^+_t:=Z_t\ind_{Z_t\ge0}$.
\end{df}

\subsection{Progressive Utilities}
To simplify the understanding of stochastic utilities and how they differ from utility functions, let recall the definition of the latter.

\noindent
A utility function is a concave strictly increasing function $U:\R\rightarrow[-\infty,+\infty)$ satisfying:
\begin{itemize}
\item The half-line $dom(U)\stackrel{def}{=}\{x\in \R;U(x)>-\infty\}$ is a nonempty subset of $\R$.
\item $U_x$ is continuous, positive and strictly decreasing on the interior of $dom(U)$, and 
\begin{equation}
U_x(+\infty)\stackrel{def}{=} \lim_{x\rightarrow +\infty}U_x(x)=0                                                                                      
\end{equation}
Set $\bar{x}:=\inf\{x\in \R; U(x)>-\infty\}$ so that $ \bar{x}\in (-\infty,+\infty)$ and either $dom(U)=(\bar{x},+\infty)$.  We define 
\begin{equation}
U_x(\bar{x})\stackrel{def}{=} \lim_{x\downarrow\bar{x}}U_x(x)                                                                                      
\end{equation}
so that $U_x(\bar{x})\in (0,+\infty]$.
\end{itemize}
In the particular case where $\bar{x}\in\{-\infty,0\}$ and $U_x(\bar{x})=+\infty$, we say that the function $U$ satisfies the Inada conditions.

\noindent
In the traditional framework, for a specified future date $T$ which is the investment horizon, an agent reflect its
preferences as a utility function, allowing it subsequently to select an optimal strategy, using the expected utility
criterion. Thus the investor will follow this strategy, which is strongly dependent on $ T $, for the future period until
maturity. Note in passing that this function is chosen independently from the investment universe and therefore can not be
adapted, in the future, to potential crises or events that may have a considerable impact on the market analysis of the
investor.

\noindent
The class of stochastic utilities $U$, studied in this paper, are also used in behavioral modeling of economic agents but
evolve dynamically in time. For this reason,  a stochastic utility $U$ is  a c\`adl\`ag\footnote[1]{Right continuous with
left hand limits.} random field \footnote[2]{
A generalization of a stochastic process such that the underlying parameter need no longer be a simple real or integer valued "time", but can instead be take values that are multidimensional vectors, or points on some manifold.}  interpreted as  a collection of $\R$-valued random variables $U(t,x)$ indexed by the time $t$ and a spacial parameter $x$ and satisfying in $x$ the classical properties of utility function. In particular, we only suppose that the first derivative exists in the classical sense, and is a continuous function. For notational simplicity, the derivative of some regular function $f$ is denoted by $f_x (x):=\frac{\partial}{\partial x}f(x)$.

\begin{df}
Given a initial utility function $U(0,x)=u(x)$, a {\em progressive utility $U$} is a c\`adl\`ag random field  $U(t,x)$ such
that the following properties hold true on a subset $\Omega^1 \in \F$ such that $\P(\Omega^1)=1$
\begin{description}
\item[(i)] For all $(t,\omega)\in[0,T_H]\times \Omega^1$ the mapping $x \mapsto U(t,x,\omega)$ from $\R$ into $\R$ is an  increasing strictly concave function  (in short utility function) of class $\Cc^1$; we also assume the positive progressive random field $U_x(t,x):=\frac{\partial}{\partial x}U(t,x)$ to be c\`adl\`ag. 
%taking finite values on all $\R$; for notational simplicity we denote $U_x(t,x):=\frac{\partial}{\partial x}U(t,x)$.
\item[(ii)] Path regularity:  For any $(t,\omega)\in[0,T_H] \times \Omega^1$ and $x\in \mathbb R$, the function $t \mapsto U(t,x,\omega)$ is c\`adl\`ag on $ [0,T_H] $ 
% \item[(iii)] {\bf Inada conditions:}
% 
% \noindent
% a) {\bf $\R$-valued wealth processes:} The random field $U(t,x)$ satisfies the  Inada conditions if  for any $(t,\omega)\in[0,T_H] \times \Omega^1$, 
% \begin{equation*}
%  U_x(t,\omega,+\infty):=\lim_{x\rightarrow +\infty}U_x(t,\omega,x)=0 \text{ and } U_x(t,\omega,-\infty):=\lim_{x\rightarrow -\infty}U_x(t,x)=+\infty \text{~a.s.}
% \end{equation*}
% 
% \noindent
% b) {\bf Positive wealth processes:} The random fields $U(t,.):\R^+ \rightarrow \R$ ( considered in  \cite{MRADNEK01}) satisfies the Inada conditions  if  for any $(t,\omega)\in[0,T_H] \times \Omega^1$, 
% \begin{equation*}
% U_x(t,\omega,+\infty):=\lim_{x\rightarrow +\infty}U_x(t,\omega,x)=0 \text{ and } U_x(t,\omega,0):=\lim_{x\rightarrow 0}U_x(t,x)=+\infty \text{~a.s.}
% \end{equation*}
 \end{description}
Finally, the random field $U(t,x)$ satisfies the  Inada conditions if  for any $(t,\omega)\in[0,T_H] \times \Omega^1$, the function $x\mapsto U(t,\omega,x)$ satisfies the  Inada conditions in the above classical sense.

\end{df}

\noindent
Obviously, this very general definition of progressive utility has to be constrained to represent, possibly changing over time, the individual preferences of an investor in a given financial market. The idea  is to calibrate these utilities with regard to some convex subclass (in particular vector space) of permitted
%positive 
 processes $X$, denoted by $\GX$, on which utilities  may have more properties.

% Finally, the set $\GX$ is homogeneous if for any $\lambda>0$, $\lambda\GX\subset \GX$.
This class $\GX$ is a general class. As the initial condition of test-processes will play a central role in this work, some subclasses of $\GX$ should be defined.  
\begin{itemize}
\item The set of all test-processes $X$ starting from the same initial condition $x$ is denoted by $\GX(x):=\{X\in \GX:~X_0=x\},~x\in \R$.
\item Let $\tau$ be a stopping time,  a $\F_\tau$-measurable random variable $\eta$ is said to be $\tau$-attainable if there exists $X\in \GX$ such that $X_\tau=\eta$ a.s.
\item A process $X$ is said to be an admissible test-process if $X\in \GX$. Furthermore,  a  process $X(\tau,\eta)$ starting at time $\tau$ from $\eta$  is said to be an admissible test-process, and we write $X(\tau,\eta)\in \GX(\tau,\eta)$, if there exists $X\in \GX$ such that $X_\tau=\eta$, and $X_{\vartheta}=X_{\vartheta}(\tau,\eta)$ a.s. for $\vartheta\geq \tau.$
\end{itemize}

\subsection{Definition of $\GX$-consistent Stochastic Utilities.}
%==================================================================================
%
We now recall the concept of consistent utilities which has been introduced by M. Musiela and T. Zariphopoulou
\cite{zar-03,zar-07} under the name "forward utilities", also called "forward performance processes".

Traditionally, the  measuring of the  performance of investment strategies by  expected utility criteria is based on  a priori specification of  a  deterministic, concave and increasing function of terminal wealth at  fixed future time.  In addition to the fact that there is no  clear idea how to specify the utility  (usually defined in isolation
to the investment opportunities) and the fact that explicit solutions to optimal investment problems can only be derived
under very restrictive model, the optimal strategy (if exists)  is strongly dependent on the investment horizon. This not
only
limits the applicability of such criteria but also poses potential inter-temporal inconsistency problems.

\noindent
Herein, an alternative that alleviates the horizon dependence, but as mentioned the notion of progressive utility on which we are interested in this paper is very large and need to be  calibrated to the convex class $\GX$. This class  is a  class of test portfolios  which  only allows to define the stochastic utility. Once his utility  defined, an investor can then turn to a portfolio optimization  problem  on the general financial market to establish his optimal strategy or  to calculate indifference prices.

\noindent
At this stage, one can ask how the class $ \GX $ is used to characterize the class of stochastic utilities? The answer is in the choice of this class and its interpretation: In finance,  $\GX$ is chosen  because  it is rather rich with high liquidity, so that  the investor is able to specify his preferences. Second, the investor have  no interest to invest  in this class and  for this reason  he  use it only to define his utility. Mathematically this latter point "no interest to invest  in this class" translates in: a supermartingale property for an arbitrary investment strategy, in other terms for any $t$-attainable wealth $X_t$ and any $X\in \GX(t,X_t)$ 
\begin{equation}
\mathbb{E}(U(s,X_s) /{\cal F}_t)\leq
U(t,X_t),  \text{~a.s.}
\end{equation} 
Finding this insufficient to characterize  stochastic utilities, we further assume that there exist a test benchmark $X^*$  for which $U(t,X^*_t)$ is a martingale. 

\noindent
Note that, this  properties: a supermartingale for an arbitrary investment strategy and a martingale at an optimum,  are also satisfied  by the value functions of the traditional problem and are a natural consequences of the dynamic programming principle.

\noindent
As the consistent utilities can be interpreted as a generalization of these value functions, in this dynamic framework, we will imposes that this properties are satisfied at any stopping time  $\tau$ starting from any $\tau$-attainable process $X_\tau$. The economic interpretation of this point is "It is never too late to optimize".
Finally in contrast to the classical framework the datum is fixed for today and not for a future time.

\begin{df}[{\bf $\GX$-consistent Utility}]\label{defUF}
A $\GX$-consistent stochastic utility  process $U(t,x)$  is a  progressive utility with the following properties:
\begin{itemize}
\item{\bf Consistency with the test-class} For  any  stopping time $\vartheta$  and  
any test  process  $X \in \GX$ s.t.  $\mathbb{E}(U(\vartheta,X_\vartheta)^+)<+\infty$, we have 

\centerline{ $\qquad \mathbb{E}(U(\vartheta,X_\vartheta) /{\cal F}_\tau)\leq
U(\tau,X_\tau)$, a.s. for any stopping time $ \vartheta\geq \tau$.}

\item {\bf Existence of benchmark process} For   any pair of stopping time and test-process $(\tau,X_{\tau})$, the constraint is saturated: that is there exists an  optimal-benchmark  process $X^*\in \GX,\text{ such that }
X^*_\tau=X_\tau$, i.e. $X^*(\tau,\eta)\in \GX(\tau,\eta)$ , and  $U(\tau,X^*_\tau)=\mathbb{E}(U(\vartheta,X^{*}_\vartheta)/{\cal F}_\tau) $  a.s. for any stopping time $ \vartheta\geq \tau$.
\end{itemize}

\noindent
In short for any test-process $X\in \GX$,
$U(t,X_t)$ is a g-supermartingale and a martingale  for the
optimal-benchmark process $X^*$.

\end{df}
In the following, the set of $\GX$-consistent stochastic  utilities will be denoted by $ {\cal U(\GX)}$ and by $\cal{ A}
\cal{U}(\GX)$ the subset of affine $\GX$-consistent stochastic utilities.

\noindent

The existence of benchmark is a strong assumption. We refer the reader to Zitkovic \cite{zitkovic} who  recently, in the case where $\GX$ is $\PX$ the set of all positive wealth processes, has taken the above property as the definition of consistent utility, by removing the assumption that the benchmark-optimal  wealth $X^*$ exists. He has found the necessary and sufficient condition under which $U$ is consistent utility. We do not consider the problem of existence of the benchmark process in this paper but as the the growth optimal portfolio (GOP) in Platen et al \cite{MR2267213}, \cite{MR2194899} properties of $X^*$ plays a crucial role in the sequel.

Note that condition $\mathbb{E}(U(\vartheta,X_\vartheta)^+)<+\infty$ leaves open the possibility that the conditional expectation $\mathbb{E}(U(t,X_t))=\mathbb{E}(U(t,X_t)^+)-\mathbb{E}(U(t,X_t)^-)$ takes the value $-\infty$ with positive probability.

\noindent
The important novel feature of our definition of consistent dynamic utilities
and this is where our notion differs from that in the work of Musiela and Zariphopoulou \cite{zar-03,zar-07},
Tehranchi et al. \cite{Mike} and Zitkovic \cite{zitkovic} is that:  First, this version of stochastic utilities is more
coherent with the financial  market in the sense that it allows, at each date $\vartheta$, to catch up and thus achieve an
optimum even if up to this date $\vartheta$ we have not made the best investment choices.
Second,  the test-processes $X$ are not required to be discounted; this variation opens the door to a more general analysis as the question of numeraire change. Third, the notion of class-test, that has not been introduced in the previous literature gives more  sense  to the notion of progressive "forward" utility, as explained above.

Note also that in the literature, consistent stochastic utilities are, in general, defined on a more large sets which are linear spaces for example $\PX$ (the set of all positive wealth processes). But one might wonder what remains to optimize after having built the utility.

\paragraph{Affine  $\GX$-consistent utilities}\label{Exists}
The purpose of this paragraph is to investigate the  affine $\GX$-consistent utilities. Note that, of course, is the simplest example of stochastic utilities but remember that any concave function is a limit of affine functions, therefore this example is very important. 

\noindent
Next result,  shows that the concept of $\GX$-consistent utilities is not vacuous and gives a sufficient condition under
which there is at least one $\GX$-consistent utility.

\begin{theo}\label{ExistsTHa}
Let $(Y_t)_t$ a positive adapted process and $(Z_t)_t$ an adapted   process, the random field $\bar{U}(t,x):=Y_tx+Z_t$ is  $\GX$-consistent utility, that
is $\bar{U}\in \cal{ A} \cal{U}(\GX)$, if and only if there exist $X^*\in\GX$ and a martingale $(M_t)_t$ such that,  
\begin{itemize}
\item[(i)]  $\bar{U}(t,x):=Y_t(x-X^*_t)+M_t$ a.s. with $M_0= Y_0X_0^*+Z_0$.
\item[(ii)] For any stopping time $\tau$ and a  $\tau$-attainable random variable $\eta$, $X^*(\tau,\eta)\in \GX(\tau,\eta)$.
\item[(iii)]  $Y_t=\bar{U}(t,X^*_t)$ a.s. and satisfies: for all $X\in \GX$, the process $\Big(Y_t\big(X_t-X^*_t\big)\Big)_{t\ge \tau}$ is a g-supermartingale and martingale for $X_.=X^*$.
\end{itemize}

\end{theo}
%Since these two results are equivalent, it suffices to show Theorem \ref{ExistsTH}.
\begin{proof}% Without loss of generality,  the initial conditions $ (\tau, \eta) $ are omitted.
Suppose that the random field $\bar{U}(t,x):=Y_tx+Z_t$ is a  $\GX$-consistent utility, then by definition there exists an
optimal process $X^*$ such that 
$\bar{U}(t,X^*_t)=Y_tX^*_t+Z_t$ is a martingale and for any test-process $X\in \GX$,
$\bar{U}(t,X_t$ is  a g-supermartingale. This implies, writing that
$$\bar{U}(t,X_t)=Y_tX_t+Z_t=Y_t(X_t-X^*_t)+Y_tX^*_t+Z_t$$ and denoting by
$M_t:=Y_tX^*_t+Z_t=\bar{U}(t,X^*_t)$, that $(M_t)_{t}$ is martingale and
$\Big(Y_t(X_t-X^*_t)\Big)_{t}$ is a g-supermartingale for any test-process $X$ which prove the
direct implication. The reverse implication is trivial.
% To show the reverse implication, it suffices to prove the "Existence of benchmark process"  as  $\bar{U}(t,x):=Y_t(x-X^*_t)+M_t$ is by definition progressive utility satisfying the Consistency  with the test class $\GX$.
\end{proof}
\begin{req}
 Assertions of Theorem \ref{ExistsTHa}, can be easy rewritten in the following dynamic version:
 \begin{itemize}
\item[(i)]    $M_\tau(\tau,\eta)= \bar{U}(\eta,\tau)=Y_\tau\eta+Z_\tau$.
\item[(ii)]  $X^*(\tau,\eta)\in \GX(\tau,\eta)$.
\item[(iii)]  $\Big(Y_t(\tau,\bar{U}(\tau,\eta))\stackrel{a.s.}{=}\bar{U}(t,X^*_t(\tau,\eta))\Big)_{t\ge \tau}$ and satisfies for all $X(\tau,\eta')\in \GX(\tau,\eta')$, the process $\Big(Y_s\big(X_s(\tau,\eta')-X^*_s(\tau,\eta)\big)\Big)_{s\ge \tau}$ is a g-supermartingale and martingale for $X(\tau,\eta)=X^*(\tau,\eta)$ a.s.
\end{itemize}
\end{req}

\noindent
 The last assertion of Theorem \ref{ExistsTHa} (equivalently $(iii)$ of the remark above) is fundamental. Existence of $ \GX
$-consistent utility requires the existence of a second process $\bar{Y}$, in addition  to an optimal test-process $\bar{X}$,
such that for any $X\in \GX$, the process
$\Big(\bar{Y}_t(X_t-\bar{X}_t))\Big)_t$ is a g-supermartingale.

\noindent
To better understand the role played by $ \bar{Y} $ and in order to successfully conclude our study, for any  stopping time $\tau$, a random variable $\eta$ $\tau$-attainable and   $\bar{X}(\tau,\eta)\in \GX(\tau,\eta)$, we denote by  $\GY_{\bar{X}(\tau,\eta)}$ and $\GY_{\bar{X}}$  the  sets given by
%  of all positive processes $Y$ such that
$$\GY_{\bar{X}(\tau,\eta)}:=\{Y\ge 0 :(Y_t(X_t(\tau,\eta)-\bar{X}_t(\tau,\eta)))_{t\ge \tau} \text{ is a  g-supermartingale, }\forall X(\tau,\eta) \in \GX(\tau,\eta) \}$$
$$\GY_{\bar{X}}:=\{Y\ge 0 : Y\in \GY_{\bar{X}(\tau,\eta)},~\forall (\tau,\eta)\}.$$

\noindent
In convex analysis, see R.T. Rockafellar \cite{Rock}, the set $ \GY_{\bar{X}}$ (resp. $\GY_{\bar{X}(\tau,\eta)}$ ) is called
the normal cone to $ \GX $ in $ \bar{X} $ (resp. to $\GX(\tau,\eta)$ in $\bar{X}(\tau, \eta)$), it is a
generalization of the concept of the dual cone. 
The reader may naturally ask the meaning of $\GY_{\bar{X}}$ ($\GY_{\bar{X}(\tau,\eta)}$), as usual the space of dual
processes do not depend on the benchmark process $ X^* $ and it's initial conditions. This dependence is mainly related to
the  structure of $ \GX$. In particular if $ \GX $ is homogeneous,  that is for any $\lambda>0$, $\lambda\GX\subset \GX$, it
is easy to see that $\GY_{\bar{X}}$ is equal to $\GY$ the set of all positive processes $Y$ such that $YX$ is a
supermartingale for all $X\in \GX$, and if $\GX$ is the set of all wealth processes  uniformly bounded by bellow, then
$\GY_{\bar{X}}$  is the set of equivalent local martingale $\M$. 

\noindent
We will see in  Section  \ref{Duality} that $ \GY $ is the analogue of $\X$  in the dual problem, but what is very
important, and we want immediately to report it  is the fact that the existence of a consistent  utility  is
strongly linked to the fact that the set $ \GY $ is empty or not. The case of linear utilities above  is a good example to
highlight this point.

\noindent
The following corollary is then a direct consequence of previous Theorems 
\begin{cor}
 $\cal{ A} \cal{U}(\GX)\neq \emptyset$ if there exist $\bar{X}\in \GX$ s.t. $\GY_{\bar{X}}\neq\emptyset$. Moreover, for any
martingale $(M_t)_t$ and any $\bar{Y}\in \GY_{\bar{X}}$ the random field $\bar{U}(t,x):=\bar{Y}_t(x-\bar{X}_t)+M_t$ is in
$\cal{ A} \cal{U}(\GX)$.
\end{cor}

% %\noindent
% %Then, one can easily show the following result
% \begin{theo}\label{ExistsTH}
% Let $(M_t)_{t\ge0}$ be any martingale. Assume that the sets $\GX$ and $\M$ are s.t. there exists a pair $\big(\bar{X},\bar{Y}\big)\in \GX\times\M$ satisfying: for any stopping time $\tau$ and any $\tau$-attainable wealth $X_\tau$, $\bar{X}(\tau,X\tau)\in \GX(\tau,X_\tau))$ and   $\big(\bar{X}_\theta(\tau,X\tau)\bar{Y}_\theta\big)_{\theta\ge\tau}$ is a martingale. Then, the affine random field $\bar{U}(t,x):=x\bar{Y}_t+M_t$ is a consistent stochastic utility.  
% \end{theo}
% 
% 
% \begin{proof}
% The proof is trivial. By definition the random field $\bar{U}(t,x)$ is a progressive utility. Moreover, for any $X\in \GX$ we have, because $\bar{Y}\in \M$ and $X$ is bounded by bellow that $XY$ is a supermartingale. Hence, as $M$ is martingale, $\bar{U}(t,X_t)$ is a supermartingale. Finally, from Assumption, for any stopping time $\tau$ and any $\tau$-attainable wealth $X_\tau$, $\bar{X}(\tau,X_\tau)\in \GX(\tau,X_\tau))$ and   $\bar{U}(\theta,\bar{X}_\theta(\tau,X_\tau))=\bar{X}_\theta(\tau,X_\tau))\bar{Y}_\theta+M_\theta,~\theta\ge \tau$ is a martingale. This shows  the result. 
% \end{proof}

% Indeed, in the sequel, we will see that this sufficient Assumption \ref{HpEx} becomes a necessary  condition to the existence of the $\GX$-consistent utilities, see corollary \ref{proproetesCor} below.

\noindent
\paragraph{Consistent utilities and value function}
An obvious question naturally arises: Is the definition of stochastic utilities do not look like a problem of optimization?

\noindent
The answer is  immediate, according to this definition, the utility process $U$ satisfies for any pair $\tau \leq \vartheta$ of stopping times,
 $$U(\tau,X^{*}_\tau)=\esssup_{X \in \GX: X_\tau=X^*_\tau }\mathbb{E}(U(\vartheta,X_{\vartheta})/{\cal
F}_\tau)\text{~a.s.} $$
The utility process is then defined from an optimization program but only on the  class $\GX$. This may seem surprising, but it is important to note that the consistent utilities $U$ are a kind of generalization of the value function $v$ of the classical portfolio optimization program which are also a solution (where $\GX$ is the set of all wealth processes  uniformly bounded by bellow) of similar identity, as it is showed by W. Schachermayer in \cite{MR2014244}.  Indeed, for a classical optimization program with maturity $T$,  the dynamic programming principle, reads as follows: for  any pair $\tau \leq \vartheta$  of $[0,T]$-valued stopping times we have 
$$v_\tau(X_\tau)=\esssup_{X \text{ admissible}: X_\tau=X^*_\tau }\mathbb{E}(v_\vartheta(X_{\vartheta})/{\cal
F}_\tau)\text{~a.s.}$$

%==========================
 \subsection{Test Processes}
Deliberately, no details on the class of test-processes  $ \GX $ is given previously, because no more is needed   to define stochastic utilities. But to carry out our study, a minimum of properties are required. 

% % \noindent
% To continue the investigations, a more precise assumptions on the structure of the set $\GX$ is also needed.
\begin{hp}\label{SB}
\rmi{\bf Convexity:}The class $\GX$ is closed and convex in the sense that is 
 $$\eps X^{1}(\tau,\eta^{1}) +(1-\eps) X^{2}(\tau,\eta^{2}) \in   \GX(\tau,\eps \eta^{1}+(1-\eps)\eta^{2})~a.s.$$  
holds for any stopping time $\tau$, any $\eta^{1},\eta^{2}$ $\tau$-admissible random variables, $X^{1}\in
\GX(\tau,\eta),~X^{2}\in
\GX(\tau,\eta')$ and   $\eps\in [0,1]$. \\
\rmii {\bf Switching property:}  For any  test-processes $X^{1}$ and $X^{2}$ in $\GX$ and 
all stopping time  $\tau$, denoting by $A_\tau$ the event $A_\tau:=\{\omega: X^{1}_\tau(\omega)=X^{2}_\tau(\omega)\} $,  the process
$\hat{X}$ defined by 
$\hat{X}_t:=X^{1}_{t\wedge\tau}+X^{1}_{t\vee\tau}\ind_{\Omega\backslash A_\tau}+X^{2}_{t\vee\tau}\ind_{A_\tau}$ is also an
element of $\GX$. 

% For any  test-processes $X^{1}$ and $X^{2}$ in $\GX(x)$,
% denoting by $\tau$ the stopping time given by $\tau:=\inf \{t:X^{1}_t\neq X^{2}_t\}$ the process  $\hat{X}$ defined by 
% $\hat{X}_t:=X^{1}_t(x)\ind_{t<\tau}+X^{2}_t(x)\ind_{t\ge\tau}$ is also in $\GX(x)$.
\end{hp}
These properties, assumed to be  satisfied by definition are a kind of guarantee to
ensure that the portfolio constraints  are  large enough and not reduced to singleton. This is very important to hope  find
solutions to our problem. We refer the reader to El Karoui \cite{ElKaroui} chap $1.$ for more details on the last point and
its role in control problems  in full generality and to \cite{kardaras} for the role of this hypothesis and its application in the financial investment optimization problem . \\
%%%%%%%%%%%
{\bf Financial Interpretation of the Set $\GX$:} A set $\GX$ that satisfies Assumption \ref{SB}  can be thought as modeling the wealth processes that
are available to some agent in a financial market. If an agent can invest at time $\tau$ in two wealth processes $X^1 \in  \GX$ and $X^2 \in \GX$, the agent should be free to allocate at time $t = 0$ a fraction $\eps\in [0,1]$ of the unit initial capital to wealth $X^1$ and the remaining fraction to the wealth $X^2$. The switching property has the following economic interpretation: if an agent can invest
in two wealth processes $X^1\in\GX $   and $X^2\in \GX$, we should then allow for the possibility that, starting with the wealth process
$X^1$, at time $\tau$ the agent decides to either switch to the wealth process $X^2$, which happens
on $A_\tau \in {\cal F_\tau}$, or keep investing according to $X^1$, on the event $\Omega \backslash A$.

%==================
\subsection{Optimality Conditions.}\label{OPCOND}
%==================
The purpose of this paragraph is to exploit the definition of consistent stochastic   utilities and to bring the properties and consequences it implies. Optimality conditions  established later in this paragraph are the key properties on which we rely to establish the main results of this paper. In particular we will show in Section \ref{NASF} that these necessary conditions are sufficient to establish the existence of stochastic utilities.

\begin{theo}[Pontryagin's Maximum Principle]\label{proprietes}
Let  $U$ be an $\GX$-consistent stochastic utility with optimal-benchmark  process $X^*$. Let $\tau$ a stopping time and a random variable $\eta$  $\tau$-attainable, then:

\noindent
If the  convex set $\GX$ is homogeneous that is for any $\lambda>0$ and any $X\in \GX$ the process  $\lambda X$ still in
$\GX$, then  
\begin{description}
\item[(i)] The process $\big(X_{t}^*(\tau,\eta)U_{x}(t,X_{t}^*(\tau,\eta))\big)_{t\ge\tau}$ is a martingale.
\item[(ii)] For any  $\tau$-attainable random variables $\eta,\eta'$, and any test-process $X\in \GX(\tau,\eta')$, the process $\big(X_{t}(\tau,\eta')U_{x}(t,X_{t}^*(\tau,\eta))\big)_{t\ge \tau}$
is a g-supermartingale.
\end{description}
\noindent
Else, $\GX$ is only assumed to satisfies Assumption \ref{SB}, 

\noindent
{\bf (OC)} For any  $\tau$-attainable random variables $\eta,\eta'$ and 
for any $X(\tau,\eta')\in \GX(\tau,\eta')$ the process
 $\big((X_{t}(\tau,\eta')-X_{t}^*(\tau,\eta))U_x(t,X_{t}^*(\tau,\eta)),~t\ge \tau\big)$ is a g-supermartingale.

\end{theo}
Before proceeding to the proof of this result, it is interesting to note that this optimality conditions established in a general way  are  quite different from those of \cite{MRADNEK01}. Indeed, in the last paper the process $U_{x}(t,X_{t}^*)_{t\ge  \tau}$ is a state density process, in turn for any test-process $X$, $X_tU_x(t,X^*_t)$ is a local martingale and a martingale if $X=X^*$. This is due essentially to the structure of the class $\GX$ which is, only, assumed to be convex in the present paper and $\GX=\PX$ (set of all positive wealth processes) in   \cite{MRADNEK01}.
% 
% Finally, before proceeding to the proof of this result let observe that this result with Theorem \ref{ExistsTH}  imply the following corollary .
% \begin{cor}\label{proproetesCor}
% If the  convex set $\GX$ is homogeneous then Assumption \ref{HpEx} is necessary and sufficient condition to the existence of the $\GX$-consistent utilities.
% \end{cor}

\begin{proof}

\noindent
To verify the above assertions observe, by convexity of $\GX$,  that for any test-process  $X(\tau,\eta')\in \GX(\tau,\eta')$ and
any $\eps\in [0,1]$, the process $\eps \big(X(\tau,\eta')-X^*(\tau,\eta)\big)+X^*(\tau,\eta')$ is a permitted  test  process  in $\GX(\tau,\eps (\eta'-\eta)+\eta)$,  starting from    $\eps (\eta'-\eta)+\eta$ at time $t=\tau$. For simplicity let us denote by $ \triangle X(\tau)$ the process given by $ \triangle X_.(\tau):=X_.(\tau,\eta')-X^*_.(\tau,\eta)$. Consequently,
by consistency property with the class $\GX$ and by martingale property of $U(.,X_.^*(\tau,\eta))$, it follows  for $\theta\ge \alpha \ge \tau$
\begin{eqnarray}\label{fehd}
&&\E\big(U\big(\theta,X_{\theta}^*(\tau,\eta)+\eps  \triangle X_\theta(\tau)\big)-U(\theta,X_{\theta}^*(\tau,\eta))/\F_\alpha\big)\nonumber \\ 
&&\le U\big(\alpha,X_{\alpha}^*(\tau,\eta)+\eps  \triangle X_\alpha(\tau)\big)-U\big(\alpha,X_{\alpha}^*(\tau,\eta)\big)\text{~a.s.}.
\end{eqnarray}

\noindent
Divide by $\eps>0$ and denote, for any $t,~\eta$ and $\eta'$, $f(\theta,.)$ the functional   
\begin{eqnarray*} 
f(\theta,\eps):=\frac{1}{\eps}\Big[U\big(\theta,X_{\theta}^*(\tau,\eta')+\eps  \triangle X_\theta(\tau)\big)-U\big(\theta,X_{\theta}^*(\tau,\eta')\big)\Big],
\end{eqnarray*}

\noindent
and observe that $\big(f(\theta,\eps),~\theta\ge \tau\big)$ is, by inequality \eqref{fehd}, a g-supermartingale satisfying, from the monotonicity of $U$, the following  $$f^+(\theta,\eps)=f(\theta,\eps)\ind_{ \triangle X_\theta(\tau)\ge 0}\text{ and }f^-(\theta,\eps)=-f(\theta,\eps)\ind_{ \triangle X_\theta(\tau)\le 0}$$

\noindent
From the derivability assumption of $U$, for any $\theta\ge \tau$, $f(\theta,\eps)$ goes to $f(\theta,0)$ when $\eps\mapsto 0$. By this, the right hand side of last inequality converge almost surely to $f(\theta,0)= \triangle X_\theta(\tau)U_x(t,X_{\theta}^*(\tau,\eta))$.
To conclude, it remains to justify the passage to the limit under the expectation. To this end,  remark that by concavity and
the increasing property of  $U(\theta,.)$ ,
$\eps\mapsto f(\theta,\eps)$ is a decreasing function with  the same  sign as $ \triangle X_\theta(\tau)$.
 Then, on the set $\{ \triangle X_\theta(\tau)\ge 0\}$, $f(\theta,\eps)$ is positive  and decreases to $f(\theta,0)$. Letting
$\eps\searrow 0$, the conditional monotone convergence theorem implies 
$$
\E\big(f^+(\theta,\eps)/\F_\alpha\big)=\E\big(f(\theta,\eps)\ind_{ \triangle X_\theta(\tau)\ge 0}/\F_\alpha\big)
\longrightarrow \E\big( f(\theta,0)\ind_{ \triangle X_\theta(\tau)\ge 0}/\F_\alpha\big)
$$

\noindent
On the other hand,  on the set $\{ \triangle X_\theta(\tau)\le 0\}$, $-f(\theta,\eps)$ is positive  and increase to
$-f(\theta,0)$.  Applying the  dominated convergence theorem, we get   for $ \theta\ge \alpha \ge \tau$
$$
\E\big(-f^-(\theta,\eps)/\F_\alpha\big)=\E\big(f(\theta,\eps)\ind_{ \triangle X_\theta(\tau)\le 0}/\F_\alpha\big)
\longrightarrow \E\big( f(\theta,0)\ind_{ \triangle X_\theta(\tau)\le 0}/\F_\alpha\big) \text{~a.s.}
$$

\noindent
This justifies the passage to the limit on the inequality  (\ref{fehd}). Hence, it follows that 
\begin{eqnarray}\label{ineq01}
\E\Big(\big(X_{\theta}(\tau,\eta')&-&X_{\theta}^*(\tau,\eta)\big)U_{x}(\theta, X_{\theta}^*(\tau,\eta))/\F_\alpha\Big)\nonumber\\ 
&\le& \big(X_{\alpha}(\tau,\eta')-X_{\alpha}^*(\tau,\eta)\big)U_{x}(\alpha, X_{\alpha}^*(\tau,\eta))\text{~a.s.}
\end{eqnarray}

\noindent
Which proves {\bf (OC)}. \\
Let, now,  focus on  the case where the  convex set  $\GX$ is homogeneous. In this case  the  stability property of $\GX$ by   positive multiplication implies that for  any $\eps >-1$, the process  $(1+\eps )X^*(\tau,\eta)\in \GX(\tau,(1+\eps )\eta)$ still permitted  and hence,
 by the same argument as above, we deduce  for  $-1<\eps<0$ respectively  $\eps >0$,  the following inequalities
\begin{eqnarray*}
\frac{1}{\eps}\E\Big(U\big(\theta,(1+\eps)X_{\theta}^*(\tau,\eta)\big)&-&U\big(\theta,X_{\theta}^*(\tau,\eta)\big)/\F_\alpha\Big)\\&\ge& \frac{1}{\eps}\Big(U\big(\alpha,(1+\eps)X_{\alpha}^*(\tau,\eta) \big)-U\big(\alpha,X_{\alpha}^*(\tau,\eta)\big)\Big) \text{~a.s.}
\end{eqnarray*}
and similarly
\begin{eqnarray*}
 \frac{1}{\eps}\E\Big(U\big(\theta,(1+\eps)X_{\theta}^*(\tau,\eta)\big)&-&U\big(\theta,X_{\theta}^*(\tau,\eta)\big)/\F_\alpha\Big)\\&\le& \frac{1}{\eps}\Big(U\big(\alpha,(1+\eps)X_{\alpha}^*(\tau,\eta) \big)-U\big(\alpha,X_{\alpha}^*(\tau,\eta)\big)\Big)\text{~a.s.}
\end{eqnarray*}

\noindent
Passing to the limit $\eps \rightarrow 0$, yields respectively
$$
\E\big(X_{\theta}^*(\tau,\eta)U_{x}(\theta, X_{\theta}^*(\tau,\eta))/\F_\alpha\big)\ge X_{\alpha}^*(\tau,\eta)U_{x}(\alpha,X_{\alpha}^*(\tau,\eta)),\text{~a.s.} ~\forall~ \theta\ge \alpha \ge \tau
$$
$$
\E\big(X_{\theta}^*(\tau,\eta)U_{x}(\theta, X_{\theta}^*(\tau,\eta))/\F_\alpha\big)\le X_{\alpha}^*(\tau,\eta)U_{x}(\alpha,X_{\alpha}^*(\tau,\eta)), \text{~a.s.}~\forall~ \theta\ge \alpha \ge \tau,
$$
then we have
$$
\E\big(X_{\theta}^*(\tau,\eta)U_{x}(\theta, X_{\theta}^*(\tau,\eta))/\F_\alpha\big)= X_{\alpha}^*(\tau,\eta)U_{x}(\alpha,X_{\alpha}^*(\tau,\eta)), \text{~a.s.}~\forall ~\theta\ge \alpha \ge \tau.
$$

\noindent
We have thus proved assertion $(i)$. Reconciling $(i)$ and {\bf (OC)} yields $(ii)$.   
\end{proof}

\subsection{Duality.}\label{Duality}
The use of convex duality in utility maximization and optimal stochastic control in general has proven extremely fruitful.
As it is established in  \cite{MRADNEK01}  analysis of utility random fields is no exception, the process $U_x(t,X^*_t),~t\ge0$ is a state price density process which is optimal to some dual problem. The idea here is  to adopt a similar approach by duality in order to prove the dual optimality of $U_x(t,X^*_t),~t\ge0$. This will   support the intuition  and  allows us a constructive intuition on different difficulties encountered in the study  of consistent progressive utilities.

\noindent
We start with a straightforward translation of the well-known Fenchel-Legendre conjugacy to
the random field case.
For a utility random field $U$ we define the dual random field $\tilde{U}: [0,+\infty[\times [0,+\infty[\times\Omega $, by
\begin{eqnarray}\label{eq:I14}
 \tilde{U}(t,y) \stackrel{def}{=}\max_{x\in \Q^*}\Big(U(t,x)-xy\Big),~{\text for }~ t\ge0,~ y\ge0 
\end{eqnarray}

\noindent
 By a simple derivation with respect to $x$,  the maximum is achieved at  $x^{*}_t=(U_x)^{-1}(t,y)=-\tU_y(t,.)$, where $(U_x)^{-1}(t,y)$ denote the inverse function of $U_x(t,.)$ with respect to the spacial parameter $x$. In turn
 \begin{eqnarray}\label{eq:I15}
 \tilde{U}(t,y) =U(t,(U_x)^{-1}(t,y))-y(U_x)^{-1}(t,y)
\end{eqnarray}

\noindent

As mentioned above, the purpose of this paragraph is to study the dual problem, for that we first specify  the set of the 
dual processes, on which one optimizes.  Unsurprising, the dual set for a given $\GX$-consistent utility with optimal
process $X^*$ is  $\GY_{X^*}$ introduced in paragraph \ref{Exists}.
%  of all positive processes $Y$ such that
From  Theorem \ref{proprietes}  $\GY_{X^*}$ is the set of potential candidates $Y$ to play the role of  $\big(U_x(t,X^*_t)\big)_t$ .

\noindent
As in the primal problem, the initial condition of the dual processes will play an important role. Then, to  formulate the
dual problem,  for  a stopping time $\tau$, a $\tau$-attainable random variable $\eta$ and $y>0$, we define  the class
$\GY_{X^*(\tau,\eta)}(\tau,y)$ by
%$$\GY_{X^*(\tau,\eta)}(\tau,y):=\{Y(\tau,y) \ge 0 :Y_\tau(\tau,y)=y,~\Big(Y_{t}(\tau,y)\big( X_{t}(\tau,\eta)-X^*_{t}(\tau,\eta)\big)\Big)_{t\ge \tau} \text{  supermartingale }\forall X(\tau,\eta) \in \GX(\tau,\eta)  \}$$
$$\GY_{X^*(\tau,\eta)}(\tau,y):=\{Y(\tau,y) \ge 0 :Y(\tau,y)\in\GY_{X^*(\tau,\eta)},~Y_\tau(\tau,y)=y \}$$
which contains, for  $y=U_x(\tau,\eta)$, the process $\big(U_x(t,X^*_{t}(\tau,\eta))\big)_{t\ge\tau}$.  

\noindent
As for test-portfolio, for a stopping time $\tau$, we introduce in the following definition the $\tau$-achievability of a
dual random variable $\kappa$. 
\begin{df}\label{kappatau}
For a stopping time $\tau$, a random variable  $\kappa$ is $\tau$-achievable if there exists a $\tau$-attainable r.v.  $\eta$ such that $\kappa=U_x(\tau,\eta)$ a.s. 
\end{df}

\noindent
The goal of this section is now, the proof of the following theorem :

\begin{theo}[Duality]\label{EDPSNHDuale'}
Let $U$ be a stochastic consistent utility with optimal-benchmark  process $X^*$. Then the convex conjugate  $\tilde{U}$  of an $\GX$-Consistent  utility $U$, given by 
(\ref{eq:I14}), satisfies 

\begin{itemize}
\item[(i)] for any $t\ge 0$, $y\mapsto \tU(t,y)$ is convex decreasing function.
\item[(ii)] for any pair $\tau \leq \vartheta$ of stopping times, for any $\tau$-attainable random variable  $\eta$  and for any $Y(\tau,\kappa)\in \GY_{X^*(\tau,\eta)}(\tau,\kappa)$, we have for $\kappa >0$ %the process $\big(\tU(t,Y_t(s,y))\big)_{t\ge s}$ is a submartingale, i.e.
\begin{eqnarray}\label{InedualeA}
\E(\tU(\vartheta,Y_\vartheta(\tau,\kappa))/\F_\tau)\ge U(\tau,\eta)-\kappa \eta+\sup_{X\in\GX(\tau,\eta)}\{\kappa \eta - \E(Y_\vartheta(\tau,\kappa)X_\theta(\tau,\eta)/\F_\tau)\},\text{~a.s.}
\end{eqnarray}
If  $\kappa$ is $\tau$-achievable with $\kappa=U_x(\tau,\eta)$, the quantity $U(\tau,\eta)-\kappa \eta$ in right side of this inequality  is replaced by $\tilde{U}(\tau,\kappa)$.

\item[(iii)] Assume the set $\GX$ to be homogeneous and  $\kappa$ to be  $\tau$-achievable with $\kappa=U_x(\tau,\eta)$ a.s. 
Then there exists a unique optimal process $Y^*_t(s,\kappa)$ s.t. %$\tilde{U}(t,Y^*_t(s,\kappa))$ is martingale, i.e.
\begin{eqnarray}\label{dualpb}
\tilde{U}(\tau,\kappa)= \E(\tilde{U}(\vartheta,Y^*_\vartheta(\tau,\kappa))/\F_\tau) = \inf_{Y(\tau,\kappa)\in \GY_{X^*(\tau,\eta)}(\tau,\kappa)} \E(\tU(\vartheta,Y_\vartheta(\tau,\kappa))/\F_\tau)\text{~a.s.}
\end{eqnarray}

Furthermore, $Y^*_\vartheta(\tau,U_{x}(\tau,\eta))=U_{x}(\vartheta,X^*_\vartheta(\tau,\eta))~ a.s. $ where we recall that $X^*_.(\tau,\eta)$ denote the optimal-benchmark  process  associated with $U$,  starting from the $\tau$-attainable capital $\eta$.
\end{itemize}
\end{theo}

\noindent
The reader should note the difference between assertions $(ii)$ and $(iii)$ of this theorem. 
Indeed, in the general case where $\GX $ is only assumed convex, the dual problem is much more complicated than the primal problem in itself. For example, we have  no idea or intuition about  the properties of processes $\big(\tU (t,Y_{t} (\tau,\kappa))\big)_{t\ge\tau} $ if they are sub or supermartingales. It is also not clear if $\tU (.,U_x(.,X^*))$ is a true martingale  or any semimartingale. Certainly the dual problem is ill posed and requires further investigation. If the set $ \GX $ is assumed convex and homogeneous, then  processes $YX$ are supermartingales and martingale for $X=X^*$ which implies that 
\begin{equation}
 \sup_{X\in\GX(\tau,\eta)}\{\kappa \eta - \E(Y_\vartheta(\tau,\kappa)X_\theta(\tau,\eta)/\F_\tau)\}=0\text{~a.s.}
\end{equation}

In this case, it is immediate that the processes $\big(\tU (t,Y_{t} (\tau,\kappa))\big)_{t\ge\tau} $ (if $\kappa$ is $\tau$-achievable ) are a submartingales and martingale for $Y=Y^*:=U_x(.,X^*)$.

Note also that the fact $\kappa$ is $\tau$-achievable plays a crucial role in this theorem. Within this assumption,
properties of submartingales and existence of an optimal dual process (in homogeneous case) are not satisfied. This is,
essentially , due to the fact that  sets $\GX(\tau,.)$ and $\GY_{X^*}(\tau,.)$    are not in perfect duality because
$(U_{x})^{-1}(.,\GY_{X^*}(\tau,.))\nsubseteq \GX(\tau,.)$, in general. In other terms, existence of solutions is intimately
related to the  inverse range of $U_{x}$,  i.e.  $(U_{x})^{-1}(.,\GY_{X^*}(\tau,.))$. For more details see \cite{KramkovAsy}
for the classical case of optimization problem. 
For example, if the range of the function $U_{x}(0,.)$ is the whole $\R^+$ (or such that asymptotic
elasticity, introduced in \cite{KramkovAsy}, is less than $1$ ) then  any $y>0$ is $0$-admissible which implies that for any
$y>0$ the dual problem (\ref{dualpb}) at $\tau=0$ (replacing $\tau$ by $0$) admits a unique solution.

\begin{req}
In the framework of \cite{MR2014244}, the identity \eqref{dualpb} is also satisfied by the convex conjugate $\tv$ of the value function of a classical optimization program with maturity $T$, that is for  any pair $\tau \leq \vartheta$  of $[0,T]$-valued stopping times, the following identity holds
$$\tilde{v}(\tau,\kappa)= \E(\tilde{v}(\vartheta,Y^*_\vartheta(\tau,\kappa))/\F_\tau) = \inf_{Y(\tau,\kappa)} \E(\tU(\vartheta,Y_\vartheta(\tau,\kappa))/\F_\tau)\text{~a.s.}$$
The proof is given in \cite{MR2014244}.
\end{req}

\begin{proof}%\marginpar{Preuve ? valider avec Nicole}
Assertion $(i)$ is a simple consequence of the definition of the convex conjugate. Let prove $(ii)$ and $(iii)$. 
 By definition of the Fenchel transform, it is immediate that for any $Y\in \GY_{X^*}$,

\begin{equation*}
 \tU(t,Y_t)\ge U(t,X^*_t)- Y_tX^*_t~~ \forall t\ge 0.
\end{equation*}

\noindent
For any $\tau$-attainable random variable $\eta$ and any test-process $X(\tau,\eta)\in \GX(\tau,\eta)$, one easily sees, using the definition  of   $\GY_{X^*(\tau,\eta)}(\tau,\kappa)$ and the martingale property of   $ (U(t,X^*_t))_{t\ge0}$, that 
\begin{eqnarray*}
\E( \tU(\vartheta,Y_\vartheta(\tau,\kappa))/\F_\tau)&\ge& \E(U(\vartheta,X^*_\vartheta(\tau,\eta))/\F_\tau)- \E(Y_\vartheta(\kappa)X^*_\vartheta(\tau,\eta)/\F_\tau)\\
&=&U(\tau,\eta)+ \E(Y_\vartheta(\tau,\kappa)\big(X_\vartheta(\tau,\eta)-X^*_\vartheta(\tau,\eta)\big)/\F_\tau)-\E(Y_\vartheta(\tau,\kappa)X_\vartheta(\tau,\eta)/\F_\tau) \\
&\ge& U(\tau,\eta)-\kappa \eta + \kappa \eta-\E(Y_\vartheta(\tau,\kappa)X_\vartheta(\tau,\eta)/\F_\tau)\text{~a.s.}
\end{eqnarray*}

\noindent
which is valid for any $X(\tau,\eta)\in \GX(\tau,\eta)$ and any $\eta$ $\tau$-attainable. Inequality \eqref{InedualeA} is then achieved by taking the supremum over $\GX(\tau,\eta)$, i.e., 
\begin{eqnarray*}
 \E(\tU(\vartheta,Y_\vartheta(\tau,\kappa))/\F_\tau)\ge U(\tau,\eta)-\kappa \eta+\sup_{X\in\GX(\tau,\eta)}\{\kappa \eta - \E(Y_\vartheta(\tau,\kappa)X_\theta(\tau,\eta)/\F_\tau)\},\text{~a.s.}~\kappa >0.
\end{eqnarray*}

Assume now that $\kappa$ is $\tau$-achievable with $\kappa=U_x(\tau,\eta)$ for some r.v. $\eta$, it follows by definition of the dual conjugate that 
\begin{equation*}
 U(\tau, \eta)-\ \kappa \eta= U(\tau,(U_x)^{-1}(\tau,\kappa))-\kappa(U_x)^{-1}(\tau,\kappa)=\tU(\tau,\kappa).
\end{equation*}

\noindent
Which proves $(ii)$. Now let turn to assertion $(iii)$. By   is homogeneity assumption of $\GX$ and  the  existence and 
$\tau$-achievability of $\kappa$ i.e.,  $U_{x}(\tau,\eta)=\kappa$ for some $\tau$-attainable r.v. $\eta$,  it follows,
denoting by $\big(X_t^*(s,\eta)\big)_{t\ge s}$ the  associated optimal process  that the process, that the process 
 $\big(Y^*_t(s,\kappa)\big)_{t\ge s}$ defined by
\begin{equation*}
Y^*_\vartheta(\tau,\kappa)=U_{x}(\vartheta,X^*_\vartheta(\tau,(U_x)^{-1}(\tau,\kappa)))>0.
\end{equation*}
is in the set $ \GY_{X^*}(\tau,\kappa)$ as by optimality conditions (Theorem \ref{proprietes}), for any $\tau$-admissible
$\eta'$ and any test-process $X\in \GX(\tau,\eta')$, the process $\big(X_{t}(\tau,\eta')Y^{*}_{t}(\tau,\kappa)\big)_{t\ge
\tau}$ is a g-supermartingale.

\noindent
Now rewriting the las identity in the following form
\begin{equation*}
(U_{x})^{-1}(\vartheta, Y^*_\vartheta(\tau,\kappa))=X_\vartheta^*(\tau,(U_{x})^{-1}(\tau,\kappa))
\end{equation*}

\noindent
which implies
\begin{equation*}
\tU(\vartheta, Y^*_\vartheta(\tau,\kappa))=U(\vartheta,X_\vartheta^*(\tau,(U_{x})^{-1}(\tau,\kappa)))-Y^*_\vartheta(\tau,\kappa)X_\vartheta^*(\tau,(U_{x})^{-1}(\tau,\kappa))).
\end{equation*}

\noindent
 One can easily deduce, since $U$ is  $\GX$-consistent stochastic utility, from the martingale property of processes
$\big(X_{t}^*(\tau,\eta)U_{x}(t,X_{t}^*(\tau,\eta))\big)_{t\ge \tau}$ and $\big(U({t},X_{t}^*(\tau,\eta))\big)_{t\ge \tau}$
and by definition of  $\big(Y^*_t(s,\kappa)\big)_{t\ge s}$,  that $\big (\tU(t, Y^*_{t}(\tau,\kappa))\big)_{t\ge
\tau}$ is also a true martingale.
Finally, for a stopping time $\vartheta\ge \tau$ using $(ii)$,

\begin{eqnarray*}%\label{ConsiDua}
\inf_{Y(\tau,\kappa)\in \GY_{X^*}(\tau,\kappa)}\E( \tU(\vartheta,Y_{\vartheta}(\tau,\kappa))/\F_\tau)&\ge& \tU(\tau,\kappa)=\E( \tU(\vartheta,Y^*_{\vartheta}(\tau,\kappa))/\F_\tau)\\
&\ge& \inf_{Y(\tau,\kappa)\in \GY_{X^*}(\tau,\kappa)}\E( \tU(\vartheta,Y_{\vartheta}(\tau,\kappa))/\F_\tau)
\end{eqnarray*}

\noindent
Which achieves the proof.
\end{proof}

\subsection{Stability by numeraire change.}\label{TCN}

We saw in the previous sections, how optimality conditions, in non-homogeneous case,  which satisfy the $\GX$-consistent  utilities  are  not intuitive. Because it is more convenient  and more simpler to work with local martingales or g-supermartingales then semimartingales, the idea of this paragraph  is to simplify the test class $\GX$, which allow us to simplify the approach and to develop a constructive  intuition about this study. More clearly, consider, for example, the context of a financial market where $ \GX $ is a class of positive wealth processes  that are semimartingales. If the set of equivalent local martingales is not empty then applying the change of numeraire $1/M$ with $M$ is an equivalent local martingale, the new wealth are positive local martingales therefore supermartingales, which is an appropriate property to the study of consistent stochastic utilities. 

\noindent
The  goal of this paragraph is then to prove the following result.

\begin{theo}[Stability by numeraire change]\label{ThSNC}
~\\
Let $U(t,x)$ be a stochastic random field and let $Y$ be a positive semimartingale, and denote by $\GX^Y$ the class of processes defined by $\GX^Y=\{\frac{X}{Y},~X\in \GX\}$, then the process $V$ defined by  
\begin{eqnarray}\label{eq:8}
V(t,x)=U(t,xY_t)
\end{eqnarray}
is  $\GX^Y$-consistent stochastic utility  if and only if $U$ is an  $\GX$-consistent stochastic utility.
\end{theo}

\noindent
Roughly speaking, the theorem says, that the notion of $\GX$-consistent  stochastic utility is preserved by 
 numeraire change. In particular, in the case of financial market, for any
 equivalent martingale measure $M$, this theorem shows that studying $\GX$-consistent stochastic utilities
is equivalent to study the $\GX^M$-consistent utilities. The advantage, here, is that the new test-processes
  in $\GX^M$ are  local martingales (in particular a supermartingales if positives). From this point, 
 we can deep the study of our utilities in the new martingale market $\GX^M$.      
\begin{proof}
 To show this result it is enough to verify assertions of  definition \ref{defUF}.
\begin{itemize}
\item[$-$] Concavity  : for
$t\geq0$, $x \mapsto V(t,x)$ is increasing concave function, by definition .

\item[$-$] Consistency with the test-class $\GX^Y$: For
any test-process  $\tX \in \GX^Y$ and any pair $\vartheta \tau$ of stopping times, $\E\big(V(\vartheta,\tX_\vartheta)=U(\vartheta,X_\vartheta)\big)<+\infty $ a.s. and
$$\qquad \mathbb{E}(V(\vartheta,\tX_\vartheta)/{\cal F}_\tau)= \mathbb{E}(U(\vartheta,X_\vartheta)/{\cal F}_\tau)\leq
U(\tau,X_\tau)\stackrel{def}{=}V(\tau,\tX_\tau) $$
%or equivalently $(u(t,X^{\pi}_t); t\geq r)$ is a \alert{supermartingale}.

\item[$-$]  Existence of optimal-benchmark: Let $\tilde{\eta}$ be a $\tau$-admissible random variable. As $U$ is  $\GX$-consistent utility and $\eta=Y_\tau\tilde{\eta}$ is $\tau$-admissible r.v. in the initial market,  there exists an optimal-benchmark  process $X^*(\tau,\eta)\in \GX(\tau,\eta)$,
$$U(\tau,\eta)=\mathbb{E}(U(\vartheta,X^{*}_\vartheta(\tau,\eta))/{\cal F}_\tau)=\esssup_{X \in \GX(\tau,\eta)}\mathbb{E}(U(\vartheta,X_{\vartheta}(\tau,\eta))/{\cal
F}_\tau),~ \forall   \tau\leq \vartheta. $$
Taking $\tX^*(\tau,\tilde{\eta})=X^*(\tau,x)/Y$ yields, by definition of $V$ and that of $\GX^Y$ we get
\begin{eqnarray*}
V(\tau,\tilde{\eta})&=&U(\tau,\eta)=\mathbb{E}(U(\vartheta,X^{*}_\vartheta(\tau,\eta))/{\cal F}_\tau)=\sup_{X \in \GX(\tau,\eta)}\mathbb{E}(U(\vartheta,X_{\vartheta}(\tau,\eta))/{\cal
F}_\tau)\\
&=&\mathbb{E}(V(\vartheta,\tX^{*}_\vartheta(\tau,\tilde{\eta}))/{\cal F}_\tau)=\sup_{\tX \in \GX^Y(\tau,\tilde{\eta})}\mathbb{E}(V(\vartheta,\tX_{\vartheta}(\tau,\tilde{\eta}))/{\cal
F}_\tau),~ \forall   \tau\leq \vartheta. 
\end{eqnarray*}
\end{itemize}
The proof is complete.
\end{proof}

\section{New approach by stochastic flows.}\label{NASF}
In this section, where $\GX$ is only assumed to be convex class, we generalize the construction of consistent progressive utilities proposed in \cite{MRADNEK01} where the market securities are modeled as a continuous semimartingale in a brownien market and where $\GX$ is the set of all positives wealth processes.   
We remind the reader that the results of the following sections  can be  stated in any class $\GX^Y$ obtained from $\GX$ by change of numeraire  and that similar results can be deduced for $\GX$  by using results of   Theorem \ref{ThSNC}. 

The main contribution of this section is the explicit construction of progressive dynamic utilities by techniques of stochastic flows composition.\\
The attentive reader might remark in the sequel that the duality approach and the duality results are not  necessary. Our
new approach is only based on the optimality conditions established in Theorem \ref{proprietes}, which we recall and analyze
in the sequel. Let begin by the main idea.

\subsection{Main Idea.}
Because we know several properties of  the  derivative  $U_x$  of an $\GX$-consistent utility $U$,
 along the optimal trajectory, i.e, $\Big(U_x\big(t, X^*_t(x)\big)\Big)_t $ given in Theorem \ref{proprietes}, the question is the following one: 
 can we  obtain  more information about the process $\big(U_x(t, x)\big)_t$, itself, from these properties?

\noindent Although this can appear too much to ask, because we try to characterize the derivative 
of a stochastic utility from its behavior on a very particular trajectory, but the answer to this
 question is positive and simple. Suppose that {\em the benchmark  process $X^*$ is strictly
 increasing with respect to its initial
 condition $x$}.  In turn the process $\big(Y^*(t,.)\big)_t$ which plays the role of $\Big(U_x\big(t,X^*_t\big((u_x)^{-1}(.)\big)\big)\Big)_t$ 
is strictly increasing with respect to $y$ because $U$ is strictly concave. Denoting by $\X(t,.)$ 
 the  inverse flow of $X^*_t(.)$,  one, easily, sees that last identity becomes,
\begin{eqnarray*}\label{eq:constutiliyflot}
 U_x(t,z)=Y^*_t(u_x(\X(t,z))),\text{~a.s.}~\forall t\ge0, z>0.
\end{eqnarray*}

\noindent
Integrating yields
\begin{eqnarray*}\label{eq:constutiliyflotA}
 U(t,x)=\int_0^xY^*_t(u_x(\X(t,z)))dz,\text{~a.s.}~\forall t\ge0, z>0.
\end{eqnarray*}

\noindent
This identity is the key of the construction propose, in this paper, in order to characterize $\GX$-consistent stochastic utilities.\\

\noindent
Note that monotony assumption  of the optimal-benchmark  process is very natural. For example,  in the results of  Example \ref{Exists}, 
the optimal benchmark process is strictly monotonous and even twice differentiable
 with respect to the initial capital $x$, under certain additional hypotheses. This is still true within the
 framework of decreasing (in the time) consistent "forward" utilities, studied by M. Musiela et al \cite{zar-08a}
 and Tehranchi et al. \cite{Mike}. We can also find these properties of the optimal process in the classic framework of
 portfolio optimization in the case of power, logarithmic, exponential utilities and in the multitude of examples 
proposed by Huy\^en Pham in \cite { Pham} and by Ioannis Karatzas and Steven Shreve in \cite { KaratzasShreve:01}.
To conclude, let us notice that, by absence of arbitrage opportunities  on the security market, the optimal process  can be only increasing with regard to the initial wealth, because otherwise by investing less money we could obtain the same gain. Mathematically, technical problems can appear, what leads to put this property as assumption.
\begin{hp}\label{CEH0}
Suppose the process $(X^*_t(x); t\ge 0)$ satisfying
\begin{eqnarray*}
 \forall t\ge 0, &x\mapsto X^*_t(x) \quad\mbox{\rm continuous and  strictly increasing, } 
  \text{\rm  s.t. }\\&X^{*}_t(-\infty)=-\infty\quad\,\> X^*_t(0)=0 \quad\,\>
 X^{*}_t(+\infty)=+\infty\text{~a.s.}
 \end{eqnarray*}
\end{hp}

\noindent
\begin{req}\label{monotonydeY}
Under this hypothesis, one can easily sees,  as the process $\big(Y^*_t(u_x(x)),~t\ge 0\big) $  plays the role of $\big(U_{x}(t,X^*_t(x)),~t\ge 0\big) $, $Y^*$ should   satisfy also,
\begin{eqnarray*}%\label{Ymonotonne}
& \forall t\ge0, x\mapsto Y^*_t(x), \text{ \rm positive strictly increasing, and s.t.  Inada conditions hold
if}\\
 &\quad Y^*_t(0)=0,\quad Y^*_t(+\infty)=+\infty\text{~a.s.}
\end{eqnarray*}
 
\end{req}

\subsection{Benchmark  process as a stochastic flow.}
The monotony assumption \ref{CEH0}  of the benchmark process $ X^*_t (x) $
 brings us naturally to consider it as the value, leaving from $x$ at $t=0$,
 of a stochastic flow $ (X^*_t (s, x))_{s\le t} $, which we define below. We can then consider 
the benchmark as leaving from condition $x$ at $t=0$ or leaving from condition $z$ at date $s$.
\begin{pp}\label{propritesduflot}
Let $ (X ^ * _ t (x)) $ be a  strictly monotonous  flow with respect to $x$ with values in $ ]-\infty, +\infty[ $.
 Its inverse $ \X (t, z) = (X^*_t(.))^{-1} (z) $ is also a strictly monotonous stochastic flow,
 defined on $  ]-\infty, +\infty[  $. We prolong the flow $X^*$ and its inverse $\X$ in the intermediate dates $ (s< t) $ in
the following way
\begin{eqnarray}\label{eq:defflot}
&X^*_t(s,x)=X^*_t(\X(s,x))\\
&\X_s(t,z)=(X^*_t(s,.))^{-1}(z)=X^*_s(\X(t,z)).\nonumber
\end{eqnarray}

\noindent
In particular, we have the following properties 
\begin{description}
\item[(i)] Equality $X_t^*(s,x)=X_t^*(\alpha,X_\alpha^*(s,x))$
hold true for all $0\le \alpha  \le  s\le t$ a.s..
\newline
Identity $\X_s(t,z)=\X_s(\alpha,\X_\alpha(t,z))$ hold true for all $0\le s \le \alpha \le t$ a.s..
\item[(ii)] Moreover, $\quad X_t^*(t,x)=x, ~\X_t(t,z)=z$, and \newline $\quad \X_s(t, X^*_t(s,x))=x, \quad
X^*_t(s,\X_s(t,x))=x,$  for all $0\le s\le t$.
\end{description}
\end{pp}

\noindent
These are important properties  which will be used several times bellow. 
For more details, we invite the reader to see H. Kunita \cite{Kunita:01}  for the general theory of stochastic flows. 
\subsection{Optimality Conditions.}
We remind in this paragraph  some results and  notations,
established in the previous section, which will play crucial role  in the sequel.
Let $U$ be an $\GX$-consistent stochastic utility, optimality conditions
imply that the derivative $U_x$ taken over the optimal-benchmark portfolio 
$X^*$, i.e. $ \big(U_x(t,X^*_t(x))\big)_t $ plays the role of dual process in  our study (Theorem \ref{proprietes}). In the case of homogeneous
constraint $ \big(U_x(t,X^*_t(x))\big)_t $ is a positive supermartingale. Furthermore, the 
process $U_x\big(t,X^*_t((u_x)^{-1}(y))\big)_t$ is the optimal dual process  of the dual optimization
problem (\ref{dualpb}) denoted by $\big(Y^*_t(y)\big)_t$ (see Theorem \ref{Duality}).   
We remind, also, that the conditions which have to satisfy
necessarily optimal processes $X^{*}(x) $ and $Y^*(y) $
as we established them in Theorem \ref{proprietes} of
paragraph \ref {OPCOND} are the following.

For  any stopping time $\tau$ and any $\tau$-attainable r.v. $\eta,~\eta'$,
\begin{description}
\item[(1)]  $X^*(\tau,\eta)\in \GX(\tau,\eta)$.
\item[(2)] For any $X(\tau,\eta')\in \GX(\tau,\eta')$,
$(X_t(\tau,\eta')-X^*_t(\tau,\eta))Y^*_t(\tau,u_x(\tau,\eta)));t\ge \tau)$ is a g-supermartingale. In other words $\Big(Y^*_t(\tau,u_x(\tau,\eta)));t\ge \tau\Big) \in \GY_{X^*(\tau,\eta)}(\tau,u_x(\tau,\eta))$.
\end{description}

% 
% \noindent
% Recall that conditions $(\mathcal{O}^*)$ are necessary conditions satisfied by the optimal-benchmark portfolio $X^*$ and  $U_{x}(.,X^*_.)$.

\noindent
From this, the monotony assumption and the above notations, it is easy to see, writing for any stopping time $\tau$ and any
r.v. $\eta$ $\tau$-attainable:  $ \eta=X^*_\tau(\X(\tau,\eta)) $, that $$ U_x(\tau, \eta) = U_x(\tau,X^*_ \tau(\X (\tau,
\eta)) = Y^*_\tau (u_x(\X (\tau, \eta)). $$ This implies, in particular, that
$$Y^*_t\big(\tau,U_x(\tau,\eta)\big)=Y^*_t\big(\tau,Y^*_\tau (u_x(\X (\tau, \eta))\big)=Y^*_t(u_x(\X (\tau, \eta)),~t\ge \tau$$
Hence, the process $\big(Y^*_t (u_x(\X (\tau, \eta))\big)_t$, starting at time $t=0$ from $u_x(\X (\tau, \eta)$,  can be interpreted as the extension to all $t\ge0$ of $\Big(Y^*_t\big(\tau,U_x(\tau,\eta)\big)\Big)_{t\ge \tau}$ which plays the role of $\big(U_x(t,X^*_t(\tau,\eta)\big)_{t\ge \tau}$.

\noindent
Summing up, from this point, optimality conditions above and the fact that the initial condition $u$ occurs in  the optimality conditions $(2)$,
 we define a set of properties to which we shall often refer afterwards. 

\begin{df}\label{def:CNoptimalite}
Let  $X^*$ and $Y^*$ be two given random fields and let $u$ an utility function. Conditions   $(\mathcal{O}^*)$ are :

\noindent
For  any stopping time $\tau$ and any $\tau$-attainable r.v. $\eta,~\eta'$,
\begin{description}
\item[(O1)]  $X^*(\tau,\eta)\in \GX(\tau,\eta)$.
\item[{\bf (OC)}] For any $X(\tau,\eta')\in \GX(\tau,\eta')$,
$(X_t(\tau,\eta')-X^*_t(\tau,\eta))Y^*_t(u_x(\X (\tau, \eta)));t\ge \tau)$ is a g-supermartingale. In other words $\big(Y^*_t (u_x(\X (\tau, \eta))\big)_{t\ge \tau} \in \GY_{X^*(\tau,\eta)}(\tau,u_x(\X (\tau, \eta)))$.
\end{description}
\end{df}

Note that, contrary to condition $(2)$, condition {\bf (OC)} is written only in terms of the initial condition $u$, the interpretation is clearer but both conditions are  equivalent.
\subsection{Construction of $\GX$-consistent utilities for a given benchmark process. }\label{Construction}
As  announced  in the introduction of this section,
our objective, under strictly monotonous hypothesis of optimal process $X^* $, is to construct $\GX$-consistent utilities of a given benchmark process $X^*$ in the class of test-processes $\GX$. After the general characterization of the consistent stochastic utilities,  the construction is presented in the special case, where the optimal dual process $Y^*$ is linear with respect to its initial condition, that is $Y^*_.(y)\equiv y\bar{Y}_t$ . This  is an interesting case because: on one side it includes the well known utilities of exponential, powers types etc ... and on the other side its gives a complete overview of the main properties of the pair $(X^*,Y^*)$ and  an intuitive explanation of the phenomenon occurring in the general case. This will be the aim of the next paragraph.
\subsubsection{Existence of $\GX$-consistent utilities for a given benchmark process.}

The previous study shows that
if, there exists $\bar{Y}\in \GY_{X^*}$ such that the process $X^*\bar{Y}$  is martingale, and not only a
 supermartingale, the  process $Y^*$ s.t. $Y^*_t(y) =y\bar{Y}_t $
is admissible in the sense that the pair $ (X^*,Y^*) $
satisfy conditions $\mathcal{O}^*$ of Definition
\ref{def:CNoptimalite} for any initial utility function $u$. 

The main idea (equation (\ref{eq:constutiliyflot})) suggests a very simple form of a $\GX$-consistent utility $U(t, x) $ of given monotonous optimal test-process. If $\X (t, z) $ denote the inverse of $X^*_t(x) $, the concave increasing  process $U(t, x) $ such that $ U_x(t, x) =u_x(\X (t, x))\bar{Y}_t $
is a good candidate  to be an $\GX$-consistent utility. Another remarkable property of this stochastic process
is that $U_x(t, X^*_ t (x)) =u_x (x)\bar{Y}_t(=Y^*_t(u_x(x)))$, what is in another way
to express that optimal dual process $Y^*_t(y)$ is linear with respect to its initial condition $y$. This is the main idea of the following result.

\begin{theo}\label{th:CNV} 
Let $X^*_t(x)$ be a test-process assumed to be strictly increasing with respect 
to the initial condition $x$ such that there exists $\bar{Y}\in \GY_{X^*}$ satisfying that the process $X^*\bar{Y}$  is martingale. Denote by $\X(t,z)$ its  inverse flow.
Then for any martingale $M$ and any utility function  $u$ such that  $u_x(\X(t,z))$ is locally integrable near $z=0$, the stochastic process
 $U$ defined by
\begin{equation}\label{eq:defutilitysimple}
U(t,x)=\bar{Y}_t\int_0^x u_x( \X(t,z))dz+ M_t
\end{equation}
is an $\GX$-consistent stochastic utility. The associated optimal  process is $X^*$ and the optimal dual process is  $Y^*(y)=y\bar{Y}$.
Further, the convex conjugate of $U$ denoted by  $\tU$, is given by 
 \begin{eqnarray}\label{deftuG}
 \tU(t,y)=\int_y^{+\infty}X^*_t(-\tu_y(\frac{z}{\bar{Y}_t})) dz +\tilde{M}_t,
\end{eqnarray}
with $\tilde{M}$ is a martingale.
\end{theo}
Note that this result generalizes the example of affine utilities given in paragraph \ref{Exists}, it suffices to take
$u_x=cte$. In particular, we stress the fact that the assumption: "There exists $\bar{Y}\in \GY$ satisfying that the process
$X^*\bar{Y}$  is martingale" is equivalent to the necessary condition at least in the homogeneous case, see assertion $(ii)$
Theorem \ref{proprietes}. This just once again highlight the necessity of optimality conditions, Theorem  \ref{proprietes},
in the study of existence of  the consistent utilities.

\medskip
\noindent
The proof of Theorem \ref{th:CNV} will be broken into several steps.

\noindent
\begin{lem}\label{lemma1}%\marginpar{La preuve marche si on remplace $X^*_s$ par $X_s$ quelconque}
 For any stopping time $\tau$, any random variable $\eta$ $\tau$-attainable  and any test-process $(X_t(\tau,\eta);s\leq t)\in \GX(s,\eta)$, we have
\begin{equation}\label{lem1eq}
\E\big(U(t,X_t(\tau,\eta)/\mathcal{F}_\tau\big)\le \E\big(U(t,X^*_t(\tau,\eta))/\mathcal{F}_\tau\big)\text{~a.s.} 
\end{equation}

\end{lem}

\noindent
\begin{proof}
By concavity of the process  $x\mapsto U(t,x)$, we have
\begin{eqnarray*}%\label{inegthA}
 U\big(t,X_t(\tau,\eta)\big)-U\big(t,X^*_t(\tau,\eta)\big)\le \big(X_t(\tau,\eta)-X^*_t(\tau,\eta)\big)
 U_x\big(t,X^*_t(\tau,\eta)\big)\text{~a.s.}
\end{eqnarray*}
\noindent 
From Definition (\ref{eq:defutilitysimple}) of $U$,  we get that $U_x\big(t,X^*_t(\tau,\eta)\big)=\bar{Y}_tu_x(\X(t,X^*_t(\tau,\eta)))$.
 On the other hand, using proposition  \ref{propritesduflot}, we have $X^*_t(\tau,\eta)=X^*_t(\X(\tau,\eta))$ and hence, by definition of $\X$,
 we obtain  $U_x\big(t,X^*_t(\tau,\eta)\big)= \bar{Y}_\tau u_x(\X(\tau,\eta))=U_{x}(\tau,\eta)\bar{Y}_{\tau,t}$ with $\bar{Y}_{s,t}:=\bar{Y}_t/\bar{Y}_s$. The inequality bellow becomes
\begin{eqnarray}\label{inegthA}
 U\big(t,X_t(\tau,\eta)\big)-U\big(t,X^*_t(\tau,\eta)\big)\le \bar{Y}_{\tau,t}\big(X_t(\tau,\eta)-X^*_t(\tau,\eta)\big)
 U_x\big(\tau,\eta\big)\text{~a.s.}
\end{eqnarray}

\noindent
We have also that $(\bar{Y}_{s,t}X^*_t(\tau,\eta),~t\ge \tau)$ is a martingale  by assumption  and   $(\bar{Y}_{\tau,t}X_t(\tau,\eta),~t\ge \tau)$ is a g-supermartingale because  $\bar{Y}\in \GY_{X^*}$. Those properties, together with (\ref{inegthA}), imply
 \begin{eqnarray*}%\label{inegthA}
 \E\Big(U\big(t,X_t(\tau,\eta)\big)-U\big(t,X^*_t(\tau,\eta)\big)/\mathcal{F}_\tau\Big)\le  \E\big(\bar{Y}_{\tau,t}(X_t(\tau,\eta)-X^*_t(\tau,\eta)\big)/\mathcal{F}_\tau\big)
 U_x\big(\tau,\eta\big)\le 0.
\end{eqnarray*}

\noindent
This will prove the validity of (\ref{lem1eq}).

\end{proof}

\begin{lem}\label{lemma2}
For any stopping time  $\tau$ and for any $\tau$-attainable random variable $\eta$, denoting $\bar{Y}_{\tau,t}:=\bar{Y}_t/\bar{Y}_\tau$ for $t\ge \tau$,
$$U(t,X^*_t(\tau,\eta))=U_x(\tau,\eta)\bar{Y}_{\tau,t}X^*_t(\tau,\eta)-\bar{Y}_t\int_{0}^{\X(\tau,\eta)}\> X^*_t(z)du_x(z)\text{~a.s.} $$
 and  it is a martingale.
\end{lem}

\noindent
\begin{proof}
Fix $\tau$ and $\eta$ any $\tau$-attainable random variable, by definition, for $t\ge \tau$, we have
\begin{eqnarray*}
 U(t,X^*_t(\tau,x))=\bar{Y}_t\int_0^{X^*_t(\tau,x)} u_x( \X(t,z))dz
\end{eqnarray*}

\noindent Consider the increasing change of  variable
$z'= \X(t,z)$ or equivalently $z=X^*_t(z')$. Using identity  $\X(t,X^*_t(\tau,x))=\X(\tau,z)$  it follows
\begin{eqnarray*}
 U(t,X_t^*(\tau,x))=\bar{Y}_t\int_{0}^{\X(\tau,x)} u_x(z)\>d_z\,X^*_t(z)%=-\int_{U_x(s,x)}^\infty z\>d_z\,X^*_t((u_x)^{-1}(z))
\end{eqnarray*}

\noindent
Integration by parts with integrability assumptions imply
\begin{eqnarray*}
 U(t,X^*_t(\tau,x))=u_x(\X(\tau,x))\bar{Y}_tX^*_t(\tau,x)-\bar{Y}_t\int_{0}^{\X(\tau,x)}\> X^*_t(z)du_x(z).
\end{eqnarray*}
Replacing $x$ by $\eta$ and using the fact that $\bar{Y}_su_x(\X(\tau,\eta))=U_x(\tau,\eta)$ yields the desired identity
\begin{eqnarray*}
 U(t,X^*_t(\tau,\eta))=U_x(\tau,\eta)\bar{Y}_{\tau,t}X^*_t(\tau,\eta)-\bar{Y}_t\int_{0}^{\X(\tau,\eta)}\> X^*_t(z)du_x(z).
\end{eqnarray*}

\noindent
While  $(\bar{Y}_{\tau,t}X^*_t(\tau,\eta),~t\ge \tau)$ is a martingale and $U_x(\tau,\eta)$ is $\F_\tau$-measurable $U_x(\tau,\eta)\bar{Y}_{\tau,t}X^*_t(\tau,\eta),~t\ge \tau$ is a martingale. Using the Fubini-Tonelli theorem,
 the integral on $u_x(z)$ of $\bar{Y}_tX^*_t(z)$  is  martingale. 
 Consequently, as a sum of two martingales, the sequence of random variables $(U(t,X^*_t(\tau,x)),~t\ge \tau)$ is a martingale.

\end{proof}

\noindent
We  now have prepared the ingredients for the   proof of Theorem \ref{th:CNV}.
\begin{proof}(Theorem \ref{th:CNV})
Since $u$ is an utility function \footnote[1]{$u$ is a strictly concave and increasing function}   and $\X$ is strictly increasing, $U(t,.)$ is a strictly concave and increasing function. To conclude, we have to check that the above Lemmas imply
 assertions $ii)$ and $iii)$ of Definition \ref{defUF}.

\noindent
Let  $(X_t(\tau,\eta);t\ge \tau)\in \GX(\tau,\eta)$ be a test-process, we have, using Lemmas \ref{lemma1} and \ref{lemma2},
\begin{eqnarray*}%\label{lem1eq}
\E\big(U(t,X_t(\tau,\eta)/\mathcal{F}_\tau\big)\le
\E\big(U(t,X^*_t(\tau,\eta))/\mathcal{F}_\tau\big) =U(\tau,\eta) \text{~a.s.}
\end{eqnarray*}

\noindent 
Which proves the consistency with the class-test $\GX$. Existence and uniqueness of optimal is a simple consequence
 of $X^*$-admissibility and strict concavity of $U$,  so that we may deduce that $U$ is an $\GX$-consistent 
stochastic utility with $X^*$ as optimal portfolio. On the other hand, the  optimal dual process is given by $U_{x}(t,X^*_t(\tau,\eta))/U_x(\tau,\eta)$  which is equal to one by construction. Finally, identity \eqref{deftuG} directly follow from the conjugacy relation $\tU_y(t,y)=-(U_x)^{-1}(t,y)$. 
\end{proof}
 \begin{req}\label{X(s,x)}
Let us note in passing that, if   the  processes $X_t(s,x)$ defined by 
\begin{equation}
 X_t(s,x)=X_t(\X(s,x))
\end{equation}
 are admissible  test-process, then we can replace $\eta$ in the previous two Lemmas,  simply, by $x$.
There is no modifications to be brought in proofs. In other words, if we can start at any time $s$ from any $x\in \R_+$ then,
 replacing $X_t(\tau,\eta)$ by $X_t(\tau,x)$ and $X^*_t(\tau,\eta)$ by $X^*(\tau,x)$, Lemmas \ref{lemma1} and \ref{lemma2} still valid. But note that this assumption suggest that any $x\in \R$ is $\tau$-attainable for any stopping time $\tau$, which is a  strong assumption. \\
Clearly, we do not make this hypothesis, which in some ways complicate our study and that of \cite{MRADNEK01}. But on the other hand, in order to overcome some difficulties the technique of stochastic change of variables is a powerful alternative for such problems.
\end{req}

\paragraph*{Application: Change of Numeraire}
It is obvious that the above theorem shows the existence of consistent utilities  and completely characterizes a large class of these random fields. This on one side, but on the other side this result accurately explains the structure of these utilities and establishes the link between these random fields and the associated  optimal processes, which is certainly true in the classical portfolio optimization by utility criterion. This message is clearer by applying the change of numeraire $1/\bar{Y}$, the following theorem rewrites as follows

%  in the case where  $\GX$ is  a class of wealth processes in a standard market, see  section \ref{FM}. Indeed, by applying the change of numeraire $1/Y$ where $Y$ is a state density process, the class $\GX^Y$ is a class of local martingales wealth processes. From this point,  the previous theorem rewrites as follows

\begin{theo}\label{th:CNV01} 
Let $\bar{X}^*_t(x)$ be a test-process in $\GX^{\bar{Y}}$ assumed to be  {\bf martingale} and  strictly increasing with respect to the initial condition. Denote by $\bar{\X}(t,z)$ its  inverse flow.
Then for any martingale $M$ and any utility function  $u$ such that  $u_x(\bar{\X}(t,z))$ is locally integrable near $z=0$,
the stochastic process  $U$ defined by
\begin{equation}\label{eq:defutilitysimpleM}
U(t,x)=\int_0^x u_x( \bar{\X}(t,z))dz+M_t
\end{equation}
is an $\GX^{\bar{Y}}$-consistent stochastic utility. The associated optimal benchmark process is $\bar{X}^*$ and the optimal dual process is constant  $\bar{Y}^*(y)=y$.
Further, the convex conjugate of $U$ denoted by  $\tU$, is given by 
 \begin{eqnarray}\label{deftuGM}
 \tU(t,y)=\int_y^{+\infty}\bar{X}^*_t(-\tu_y(z)) dz+\tilde{M}_t,
\end{eqnarray}
with $\tilde{M}$ is a martingale.
\end{theo}

\noindent
 This result merits some comments. First, the derivative of the stochastic utility is other than a deterministic function of
the inverse map of the optimal portfolio, equivalently the  derivative $\tU_y(t,y)$ of the convex conjugate is exactly
(minus) the optimal benchmark (the optimal wealth in the case of financial market) with the initial condition
$(u_x)^{-1}(y)$.
Second, starting from a financial market the martingale market, obtained by change of numeraire $1/Y$ with $Y$ is a State price density process, is not unique. Then the fact that the optimal dual process is constant does not mean that the market is complete but that this dual optimal process is linear with respect to its initial condition $u_x(x)$ in the selected martingale market.% (see Theorem \ref{th:CNVAAA} for more intuitive interpretation).  

\subsubsection{Construction of all $\GX$-consistent utilities for a given  benchmark process.}\label{ConstructionG}
%%%%%%%%%%%%%%%%%%%%%%%%%%%%%%%%%%%%%%%%%%%%%%%%%%%%%%%%
In this section, we turn to the central result of this paper. We showed  in Theorem \ref{th:CNV} that for any increasing  test-process $X^* $, such $X^*$  is a 
martingale, we can construct a consistent utilities of optimal benchmark process
 $X^* $. The feature of these consistent utilities, defined by (\ref{eq:defutilitysimple}), 
is that the optimal dual process is fixed to $1$. 
In order to characterize all consistent utilities with given optimal portfolio $X^*$,
 we consider  more general class of processes $Y^*$ such that
 optimality conditions
$\mathcal {O}^* $ are satisfied for the pair
$ (X^*,Y^*)$. As we saw it, the intuition is to characterize
utilities  $U$ such that $U_{x}(t, x) = Y^*\big(u_x( \X (t, x))\big) $, where $\X (t, x)$
is the inverse flow of $X^*$. The monotony condition   of $X^*$
draw away that the stochastic flow $Y^*$ must be increasing to
guarantee that $U_{x}(t, x) $ is decreasing. To resume, in the sequel we, only, consider pairs $(X^*,Y^*)$
  of  processes and utility function $u$ satisfying 

\begin{hp}\label{hpC1}
\begin{description}
 \item[(\small{A1})] The process $(X^*_t(x);x\in \R,t\ge
0)$ is strictly increasing from $-\infty$ to $+\infty$ while  $(Y^*_t(y);~y\ge 0,t\ge 0)$, according to remark \ref{monotonydeY}, is strictly
 increasing from $+\infty$ to $0$ such that $Y^*_t(u_x(x))$ is locally integrable near $x=0$.
% $\Y(0,x)$ is a strictly decreasing  functional, denoted by $u_x(x)$.
 \item[(\small{A2})] The triplet $ (X^*,Y^*,u)$ satisfy $\mathcal {O}^* $.
\end{description}
\end{hp}

\noindent
The martingale property  of the process $\big(X^{*}_t(x)Y^*_t(y); t\ge 0\big)$ played a key role in establishing the validity of Lemma \ref{lemma2}  and consequently that of Theorem \ref{th:CNV}. In general this property is not satisfied  (Theorem \ref{proprietes}). However, one can remark that, in general case, that the martingale property  hold true with $X^{*}$ replaced by his derivative  $D_xX^*$  with respect to $x$ (if it exists). Indeed, from Theorem \ref{proprietes}, for any $\delta>0$
$$\Big(\big(X^*_t(x+\delta)-X^*_t(x)\big)Y^*_t(u_x(x))\Big)_t ~\text{ is g-supermartingale}$$
$$\Big(\big(X^*_t(x-\delta)-X^*_t(x)\big)Y^*_t(u_x(x))\Big)_t ~\text{ is g-supermartingale}$$
If  $D_xX^{x}$ exists, one gets, letting $\eps\searrow 0$, that 
$\mp Y^*(u_x(x))D_xX^*(x)$ is a g-supermartingale, then martingale. In the following, this property implies a generalization
of  Lemma \ref{lemma2} which will  be needed to show the main result of this work.

\noindent
To justify passage on the limit and furthermore, in order to generalize our new approach (Theorem \ref{th:CNV}), the following domination assumption suffices.

\begin{hp}\label{hpIntegrabiliteX^*}
 \begin{description}
 \item[H1 local)] For all $x$, there exists an integrable positive adapted process,
  $U_t(x)>0$ such that, if we denote by  $\mathbf{B}(x,\alpha)$ the ball of radius $\alpha>0$ centered at $x$,
\begin{eqnarray}
 \forall y,y' \in \mathbf{B}(x,\alpha), |X^*_t(y)-X^*_t(y')|<|y-y'|\,U_t(x),\text{~a.s.} ~t\ge 0
\end{eqnarray}
 \item[ H2 global) ] $U_t(x)$ is increasing with respect to $x$ and $U^I_t(x)=\int_0^x\Y(t,z)U_t(z)dz$ is integrable for all $t\ge 0$.
\end{description}
\end{hp}

\noindent
Let us point out that this hypothesis is introduced only to justify result of the following proposition. Summing up, under this assumption

\begin{pp}\label{Yd_zXmartingale}
Let assumptions \ref{hpC1} and \ref{hpIntegrabiliteX^*} hold. If the derivative with respect to $x$
of the increasing  process $X^*_t(x)$ denoted by $D_x\,X^*_t(x)$ exists in any point $x$, then
$Y^*_t(u_x(x))D_x\,X^*_t(x)$ is a martingale.
Otherwise, without derivability assumption, the process
\begin{eqnarray}\label{integralA}
 \int_0^xY^*_t(u_x(z))d_zX^*_t(z),
\end{eqnarray}

\noindent
is also a martingale.
\end{pp}

\medskip
\noindent
We show  in the proof of Theorem \ref{th:CNV}, that quantity 
\begin{eqnarray*}
 \int_0^{\X(\tau,\eta)}Y^*_t(u_x(z)) d_zX^*_t(z).
\end{eqnarray*}

\noindent
corresponds to $U (t,X^*_t(\tau,\eta)) $  where $U$ is a process which we define afterwards. Particularly,
  this proposition is other than a generalization of Lemma \ref{lemma2} where we replace deterministic
 quantity $u_x $ by the process $\Y$.

\begin{proof} 
 For the duration of the proof we write $\Y(t,x)$ for $Y^*_t(u_x(x))$. Then we have to show that 
\begin{eqnarray*}
 \int_0^{\X(\tau,\eta)}\Y(t,z) d_zX^*_t(z).
\end{eqnarray*}
is a martingale.

{\bf a)} First, suppose  $X^*_t(x)$  is differentiable with respect to  $x$. For
$0<\epsilon<\alpha$, the process
$\Y(t,x)\big(X^*_t(x+\epsilon)-X^*_t(x)\big)$ is a positive supermartingale (assertion ({\bf(OC)} of Theorem \ref{proprietes}). By assumption 
\ref{hpIntegrabiliteX^*} the right derivative with respect to $\epsilon$,  $\Y(t,x)D^+_x\,X^*_t(x)$ is a  positive supermartingale.\\

On the other hand, $\Y(t,x)\big(X^*_t(x)-X^*_t(x-\epsilon)\big)$ is a positive submartingale. Once again, the hypothesis \ref {hpIntegrabiliteX^*} is
 used to show that we can again pass to the limit and deduct that $\Y(t,x)D^-_xX^*_t(x)$ is a positive submartingale. From derivability of
 $X^*$, $D^-_xX^*_t(x)=D^+_xX^*_t(x)=D_xX^*_t(x)$ and then the process $\Y(t,x)D_xX^*_t(x)$ is, consequently,  a sub and  supermartingale and therefore martingale.\\

\noindent
{\bf b)} In the general case, without differentiability assumption on $X^*$, we use  Darboux sum to study the properties of
$S(x)=\int_0^x\Y(t,z)d_zX^*_t(z)$. We partition the interval  $[0,x]$ into $N$ subintervals   $]z_n,z_{n+1}]$ where the mesh approaches zero. 
To approach the integral (\ref{integralA}) by below respectively by above we consider respectively the following sequences
\begin{eqnarray*}
S_N(t,x)&=&\sum_{n=0}^{n=N-1}\Y(t,z_{n})\big(X^*_t(z_{n+1})-X^*_t(z_{n})\big)\\
 S'_N(t,x)&=&\sum_{n=0}^{n=N-1}\Y(t,z_{n+1})\big(X^*_t(z_{n+1})-X^*_t(z_{n})\big).\\
\end{eqnarray*}

\noindent
By the same arguments as above, the sequence $S_N (t, x) $ is a positive supermartingale, 
while the sequence  $S ' _N (t, x) $ is a positive submartingale, and a positive local martingale if $\GX$ is homogeneous. 
In all cases, by hypothesis \ref {hpIntegrabiliteX^*}, the positive processes $ S_N (t, x) $ and $S'_N (t, x) $ are bounded above by $$\bar{S}_N(t,x):=\sum _ {n=0}^{n=N-1} \Y (t, z_{ n+1 }) U_t (z_{ n+1 }) $$ 
Moreover, under assertion H2 global) of hypothesis \ref {hpIntegrabiliteX^*} , $\bar{S}_N(t,x)$ is bounded above by $U^I_t(x) = \int_0^x\Y (t, z) U_t ( z )dz$. As the properties of sub and  supermartingale  are preserved in passing to the limit it follows that $\int_0^x\Y (t, z) d_zX^*_t (x)$ is  a martingale.
\end{proof}

\noindent
 We have now all  elements to characterize consistent utilities of given optimal benchmark.

\begin{theo}[General Characterization]\label{thP'}
Let $(X^*,Y^*)$  be a pair of  processes and $u$ any utility function such that  assumptions \ref{hpC1} and \ref{hpIntegrabiliteX^*} hold.
Let $\X$ the inverse flow of $X^*$, $\Y$ the inverse flow of $Y^*$, $M$ a martingale  and $\tu$ the convex conjugate of $u$. Then the concave increasing process $U$ defined by
\begin{equation}\label{defuG}
 U(t,x)=\int_0^x Y^*_t\big(u_x(\X(t,z))\big)dz+M_t
\end{equation}
is an $\GX$-consistent stochastic utility  with 
$u$ as the initial function,  $X^*$ as the  optimal benchmark process. The optimal dual process is $Y^*$ and the convex conjugate is given by
\begin{eqnarray}\label{defuGD}
\tilde{U}(t,y)=\int_y^{+\infty} X^*_t\Big(-\tu_y\big(\Y^*(t,z)\big)\Big)dz+\tilde{M}_t.
\end{eqnarray}
With $\tilde{M}$ is a martingale.
\end{theo}

\noindent
In Theorem \ref{th:CNV}, for a given initial utility, we construct an $\GX$-consistent utility of given optimal portfolio (martingale). The extension which we give here which, up-technical points, characterizes all the $\GX$-consistent utilities equivalent to the previous one (in the sense that they gives the same optimal portfolio process). This characterization expresses only how we have to diffuse the function $u_x(x) $ to stay within the framework of the $\GX$-consistent utilities. The answer is intuitive because it expresses that it is enough to keep a monotonous  flow  $Y^*\in \GY_{X^*}$: $Y(X-X^*),~X\in \GX$ are a g-supermartingale. Moreover, note that in the forward problem the idea, at the beginning, is to diffuse the initial utility $u$ using the information given by the path of $X^*$.  Contrary to what one might think we observe clearly that the diffusion is not on $u$ but on the derivative $u_x$.

\begin{proof}

As in the  previous Theorem, the proof is made in two step.
 The consistency with the universe of investment is based on two essential properties:\\
 $-$ On one hand on the fact that $(U(t, X^*_t (s,\eta)),~t\ge s)$  is  a martingale.\\
$-$ On the other hand, the consistency with the class-test $\GX(s,\eta)$.

To show these properties, we begin by the following result which is the extension of Lemma \ref{lemma1}.

\begin{lem}%\marginpar{La preuve marche si on remplace $X^*_s$ par $X_s$ quelconque}
Under assumptions of the previous theorem, for any stopping time $\tau$, any random variable $\eta$ $\tau$-attainable  and any test-process $(X_t(\tau,\eta);s\leq t)\in \GX(s,\eta)$, we have
\begin{equation}\label{lem1eq2}
\E\big(U(t,X_t(\tau,\eta)/\mathcal{F}_\tau\big)\le \E\big(U(t,X^*_t(\tau,\eta))/\mathcal{F}_\tau\big)\text{~a.s.} 
\end{equation}

\end{lem}

\noindent
\begin{proof} The proof is identical to that of Lemma \ref{lemma1}.
By concavity of the process  $x\mapsto U(t,x)$, it follows 
\begin{eqnarray*}%\label{inegthA}
U\big(t,X_t(\tau,\eta)\big)-U\big(t,X^*_t(\tau,\eta)\big)\le \big(X_t(\tau,\eta)-X^*_t(\tau,\eta)\big)
U_x\big(t,X^*_t(\tau,\eta)\big)\text{~a.s.}
\end{eqnarray*}
\noindent 
By Definition  of $U$ and $X^*_t(\tau,\eta)=X^*_t(\X(\tau,\eta))$,   $U_x\big(t,X^*_t(\tau,\eta)\big)=Y^*_t\big(u_x(\X(t,X^*_t(\tau,\eta)))\big)=Y^*_t\big(\tau,U_x(\tau,\eta)\big)$.
 The inequality bellow becomes
\begin{eqnarray}\label{inegthA2}
 U\big(t,X_t(\tau,\eta)\big)-U\big(t,X^*_t(\tau,\eta)\big)\le Y^*_t\big(\tau,U_x(\tau,\eta)\big)\big(X_t(\tau,\eta)-X^*_t(\tau,\eta)\big) \text{~a.s.}
\end{eqnarray}

\noindent
By Assumption,  for any $X(\tau,\eta)\in \GX(\tau,\eta)$,  $Y^*_t\big(\tau,U_x(\tau,\eta)\big)\big(X_t(\tau,\eta)-X^*_t(\tau,\eta)\big),~t\ge \tau)$ is a g-supermartingale. Those properties, together with (\ref{inegthA2}), imply
 \begin{eqnarray*}%\label{inegthA}
 \E\Big(U\big(t,X_t(\tau,\eta)\big)-U\big(t,X^*_t(\tau,\eta)\big)/\mathcal{F}_\tau\Big)\le  \E\big(Y^*_t\big(\tau,U_x(\tau,\eta)\big)\big(X_t(\tau,\eta)-X^*_t(\tau,\eta)\big)/\mathcal{F}_\tau\big)\le 0.
\end{eqnarray*}

\noindent
This will prove the validity of (\ref{lem1eq2}).

\end{proof}
%%%%%%%%%%%%%%%%%%%%%%%%%%%%%%%

\noindent 
To conclude, it suffices to  show that $U(t,X^*_t(x))
 $ is a martingale. To be made, we proceed as in Lemma \ref{lemma2} by writing that 
 $U(t,X^{*}_t(s,\eta))=\int_0^{X^*_t(s,\eta)}Y^*_t\Big(u_x\big(\X(t,z')\big)\Big)dz'$. %= \int_0^{\infty}\Y o\X(t,z')\ind_{\{0\le z'\le X^{*}_t(s,\eta)\}}dz'$.\\
Let us make the  change of variable $\X(t,z') =z$, consequently, because $\X(t,X_t^*(s,\eta))=\X(s,\eta)$, we get 

\begin{eqnarray*}
U(t,X^{*}_t(s,\eta))%&=&\int_0^{\X(s,\eta)}\Y(t,z)\ind_{\{0\le X^{*}_t(z)\le X^{*}_t(s,\eta)\}}d_z(X^{*}_t(s,z))
=\int_0^{\X(s,\eta)}Y^*_t\Big(u_x(z)\Big)d_z(X^{*}_t(z))
\end{eqnarray*}

\noindent
Finally, by proposition \ref{Yd_zXmartingale},  $\Big(\int_0^{\X(s,\eta)}Y^*_t\Big(u_x(z)\Big)d_z(X^{*}_t(z)),~t\ge s\Big)$ is a martingale and,
 hence, $\Big(U(t,X^{*}_t(s,\eta)),~t\ge s\Big)$ is a martingale and the proof is complete.
\end{proof}

\noindent

The general characterization of consistent stochastic utility in this result is given  according to the initial condition at
time $0$, $U_x(0,.)=u_x$. But it is also possible to write the formula for any intermediate date  $s$ as it is given in the
following result

\begin{cor}%[Second General Characterization]\label{thP''}
Under assumptions of Theorem \ref{thP'}, for any stopping time $\tau$, the $\GX$-consistent stochastic utility  $U$ defined by \eqref{defuG} and its convex conjugate $\tU$ are rewritten
\begin{eqnarray*}
 U(t,x)&=&\int_0^x Y^*_t\Big(\tau,U_x(\tau,\X_\tau(t,z))\big)dz,~t\ge \tau\\
\tilde{U}(t,y)&=&\int_y^{+\infty} X^*_t\Big(-\tU_y\big(\tau,\Y_\tau(t,z)\big)\Big)dz,~t\ge \tau.
\end{eqnarray*}
\end{cor}

\begin{proof}
Recalling the notation $\X_0(\tau,x)=\X(\tau,x)$, the proof of this result is based on the previous theorem. Indeed,   rewriting for $t\ge \tau$ $$U_x(t,x)=Y^*_t(u_x(\X(t,x)))=Y^*_t\Big(\tau,Y^*_\tau\big(u_x(\X(t,x))\big)\Big)$$ and using the fact that 
$\X\big(\tau,\X_\tau(t,x)\big)=\X(t,x)$ it follows 
\begin{eqnarray*}
U_x(t,x)&=&Y^*_t\Big(\tau,Y^*_s\big(u_x(\X\big(\tau,\X_\tau(t,x)\big))\big)\Big)\\
        &=&Y^*_t\Big(\tau,\big[Y^*_\tau\big(u_x(\X\big(\tau,.\big))\big)\big](\X_\tau(t,x))\Big)
\end{eqnarray*}

From this point and the identity $U_x(\tau,.)=Y^*_\tau\big(u_x(\X(\tau,.))\big)$ yields
\begin{eqnarray*}
U_x(t,x)&=&Y^*_t\Big(\tau,\big[Y^*_\tau\big(u_x(\X\big(\tau,.\big))\big)\big](\X_\tau(t,x))\Big)\\
        &=&Y^*_t\Big(\tau,U_x\big(\tau,\X_\tau(t,x)\big)\Big)        
\end{eqnarray*}
Integrating yields the result. Inverting the roles of $X^*$ and $Y^*$,  the same arguments allows us to establish the dual identity.
\end{proof}

%%%%%%%%%%%%%%%%%%%%%%%%%%%%%%%%%%%%%%%%%%%%%%%%%%%%%%%%%%%%%%%%%%%%%%%%%%%%%%%
%%%%%%%%%%%%%%%%%%%%%%%%%%%%%%%%%%%%%%%%%%%%%%%%%%%%%%%%%%%%%%%%%%%%%%%%%%%%%%%

% \section{Financial Market: Examples of $\GX$}\label{FM}
% 
% \noindent
% {\cgblue As we presented in the introduction, our purpose is to define a stochastic utilities that we use to compare investment opportunities in a given financial market. We distinguish two types of strategies, those on what we will optimize and those on what we know a priori that it is not necessary to invest, because the associated performance is too low, however they play a very important role to calibrate the utility processes.
% }

%%%%%%%%%%%%%%%%%%%%%%%%%%%%%%%%%%%%%%%%%%%%%%%%%%%%%%%%%%%%%%%%%%%%%%%%%%%%%%%
%%%%%%%%%%%%%%%%%%%%%%%%%%%%%%%%%%%%%%%%%%%%%%%%%%%%%%%%%%%%%%%%%%%%%%%%%%%%%%%

\subsection*{Conclusion}
 Despite the abstract framework of this paper and although the results here are, under minimal regularity assumptions, an extension of those set in \cite{MRADNEK01}, the proofs in this work are much simpler and require less computations. Second, as  announced at the beginning of this work the results and the method are valid for  more general convex sets of test-processes, provided they are rich enough, because only the property of convexity plays a role in the proofs of Theorems.  Finally, as we have seen we can do without the duality, which allows an interpretation of  process $U_x(t,X^*_t)$.

\nocite{*}
\bibliographystyle{plain}
\bibliography{ArefT}
\end{document}